\documentclass{ws-jaa}

\usepackage{xcolor}
\usepackage[sort,compress]{cite}
\usepackage[verbose]{hyperref}
\hypersetup{colorlinks=false,allbordercolors=blue,pdfborderstyle={/S/U/W 1}}

\usepackage{stmaryrd}
\SetSymbolFont{stmry}{bold}{U}{stmry}{m}{n}

\usepackage[T1]{fontenc}
\usepackage{mathtools}
\usepackage{tikz}
\usepackage{graphicx}
\usepackage{pgfplots}
\usepackage{subcaption}
\usepackage{stmaryrd}
\usepackage{microtype}
\usepackage{enumerate}
\usepackage{array}
\usepackage[all]{xy}

\definecolor{BC}{HTML}{0000aa}
\definecolor{GV}{HTML}{aa0000}

\makeatletter
\@namedef{subjclassname@2020}{%
	\textup{2020} Mathematics Subject Classification}
\makeatother

\makeatletter
\newcommand{\customlabel}[2]{%
   \protected@write \@auxout {}{\string \newlabel {#1}{{#2}{\thepage}{#2}{#1}{}} }%
   \hypertarget{#1}{#2}%
}
\newcommand{\customlabelquiet}[2]{%
   \protected@write \@auxout {}{\string \newlabel {#1}{{#2}{\thepage}{#2}{#1}{}} }%
   \hypertarget{#1}{}%
}
\makeatother

\usepackage{cleveref}
\renewcommand{\tilde}{\widetilde}
\renewcommand\epsilon{\varepsilon}

\newcommand\extlag{n}
\newcommand\F{\mathbb{F}_q}
\newcommand\LL{\mathbb{F}_{q^r}}
\newcommand\Fq{\mathbb{F}_q}
\newcommand\Fn{\mathbb{F}_{q^\extlag}}
\newcommand\LLn{\mathbb{F}_{q^{\extlag r}}}
\newcommand\R{\mathbb{R}}

\newcommand\Q{\mathbb{Q}}
\newcommand\Z{\mathbb{Z}}
\newcommand\N{\mathbb{N}}

\newcommand\m{\mathfrak{m}}

\newcommand\BB{\mathcal{B}}

\newcommand\CC{\mathcal{C}}
\newcommand\HH{\mathcal{H}}

\newcommand\Gal{\mathrm{Gal}}
\newcommand\End{\mathrm{End}}
\newcommand\Hom{\mathrm{Hom}}

\newcommand\inv{\mathrm{inv}}

\newcommand\ua{\boldsymbol{a}}
\newcommand\ub{\boldsymbol{b}}

\newcommand\uu{\boldsymbol{u}}
\newcommand\uv{\boldsymbol{v}}
\newcommand\uw{\boldsymbol{w}}
\newcommand\ue{\boldsymbol{e}}
\newcommand\uX{\boldsymbol{X}}
\newcommand\uZ{\boldsymbol{Z}}
\newcommand\utheta{\boldsymbol{\theta}}
 \newcommand\rk{\mathrm{rank}}
\newcommand\srk{\mathrm{srk}}
\newcommand\ie{\textit{i.e.~}}
\newcommand\eg{\textit{e.g.~}}
 \newcommand\mindist{d}
  \newcommand\degpol{c}
 \newcommand\numvar{m}
\newcommand\rgcd{\text{\rm rgcd}}
\newcommand\Nrd{N_{\text{\rm rd}}}
\newcommand\ord{\text{\rm ord}}
\newcommand\Homgrp{\text{\rm Hom}_{\text{\rm grp}}}
\newcommand\LAG{\text{\rm LAG}}
\newcommand\LRM{\text{\rm LRM}}
\newcommand\Card{\text{\rm Card}}
\newcommand\Supp{\text{\rm Supp}}

\title{Reed--Muller codes in the sum-rank metric}

\author{Elena Berardini}
\address{CNRS; IMB, Université de Bordeaux, 351 cours de la Libération\\33405 Talence, France\\
\email{elena.berardini@math.u-bordeaux.fr}}

\author{Xavier Caruso}
\address{CNRS; IMB, Université de Bordeaux, 351 cours de la Libération\\33405 Talence, France\\
\email{xavier.caruso@normalesup.org}}

\begin{document}

\maketitle
\begin{abstract}
We introduce the sum-rank metric analogue of Reed--Muller codes, which we called linearized Reed--Muller codes, using multivariate Ore polynomials. We study the parameters of these codes, compute their dimension and give a lower bound  for their minimum distance. Our codes exhibit quite good parameters, respecting a similar bound to Reed--Muller codes in the Hamming metric. Finally, we also show that many of the newly introduced linearized Reed--Muller codes can be embedded in some linearized Algebraic Geometry codes, recently defined in \cite{BC24}, a property which could turn out to be useful in light of decoding.

\end{abstract}

\bigskip

\begin{center} \emph{Dedicated to Sudhir Ghorpade for his 60$^{\text{th}}$ birthday.}
\end{center}
\ccode{2020 Mathematics Subject Classification: 11T71, 94B05, 16U20}

\section{Introduction}

Error-correcting codes are powerful tools for securing data transmission 
over unreliable and noisy channels. Traditionally, one wants to protect data
against bit erasure or bit flipping, which leads to consider the so-called 
Hamming distance. Nonetheless, for some specific applications (\eg 
transmission of a message through a network in which some 
servers can be down or malicious) another metric is more relevant 
than the classical Hamming one: it is the rank metric.
Interpolating between those, one finds the sum-rank metric; it
was recently introduced and now finds applications 
in many areas of information theory, such as 
multi-shot linear network coding, space-time coding, and distributed 
storage systems (one can consult \eg \cite{MSK22} for an overview of all 
these applications), and consequently have attracted significant 
attention of researchers from different fields. However, in contrast 
with the situation of codes in the other aforementioned metrics, 
particularly the Hamming one, only a few constructions of codes in 
the sum-rank metric are known and have been thoroughly studied.

Let $\F$ be a finite field with $q$ elements, and let $\LL$ be an extension of degree~$r$ of it. Classically, codes in the sum-rank metric are defined as subspaces of the product of some spaces of matrices with coefficients in $\Fq$. In particular, one can consider both $\F$-linear and $\LL$-linear subspaces. In the present paper, we will only deal with $\LL$-linear codes, and take the point of view of spaces of endomorphisms rather than spaces of matrices. In this context, sum-rank metric codes are defined as follows. For an integer $s$, set 
\[\HH :=  \prod_{s} \End_{\F}(\LL).\]
This is a vector space over $\LL$, of dimension $s r$. Let $\boldsymbol{\varphi}=(\varphi_1,\dots,\varphi_s)\in\HH$. The sum-rank weight of $\boldsymbol{\varphi}$ is defined as \[w_{\srk}(\boldsymbol{\varphi})\coloneqq\sum_{i=1}^s \rk (\varphi_i)=\sum_{i=1}^s \dim_{\F} \varphi_i(\LL).\]
The sum-rank distance between $\varphi$ and $\psi\in \HH$ is 

\[d_{\srk}(\boldsymbol{\varphi},\boldsymbol{\psi})\coloneqq w_{\srk}(\boldsymbol{\varphi}-\boldsymbol{\psi}). \]

 \begin{definition}
A ($\LL$--linear) code $\mathcal{C}$ in the sum-rank metric is a $\LL$--linear subspace of 
$\HH$ endowed with the sum-rank distance.
By definition, its \emph{length} $n$ is $\dim_{\LL} \HH = s r$. Its 
\emph{dimension} $k$ is $\dim_{\LL} \mathcal{C}$. Its \emph{minimum 
distance} is 
\[\mindist\coloneqq\min \left\{w_{\srk}(\boldsymbol{\varphi}) \mid 
\boldsymbol{\varphi} \in \mathcal{C}, \boldsymbol{\varphi}\neq 
\boldsymbol{0}\right\}.\]
\end{definition}
The three main parameters of a code in the sum-rank metric are related by the equivalent of the 
Singleton bound in the Hamming metric, that in the aforementioned 
setting reads $\mindist + k \leq n + 1$ \cite[Proposition 34]{MP18}. Codes
with parameters attaining this bound are called \emph{Maximum Sum-Rank 
Distance (MSRD)}. 
Let us mention that if $r =1$, the previous definition 
reduces to codes of length $s$ with the Hamming metric and, if $s = 1$, to rank-metric codes. This highlights why the sum-rank metric is considered as a generalization of both metrics. For a comprehensive overview on sum-rank metric codes we refer the reader to \cite{GMPS23}.
\smallskip

As for rank-metric codes, a central question in the study of  sum-rank metric codes is to find constructions analogue to the existing ones in the Hamming metric. The counterpart of Reed--Solomon codes in the sum-rank metric are the so-called linearized Reed--Solomon codes \cite{MP18}, whose construction relies on the use of Ore polynomials. Algebraic Geometry codes in the sum-rank metric were recently introduced by the authors \cite{BC24}, using again Ore polynomials but with coefficients in the function field of an algebraic curve. Among the most used families of linear codes in the Hamming metric, Reed--Muller codes \cite{Muller54,Reed54} constitute a widely studied class which however does not have its analogue in the sum-rank metric yet.  The main goal of the present paper is to fill this gap. 

\subsection*{Our contribution} In this paper we present the first analogue of Reed--Muller codes in the sum-rank metric, that we call linearized Reed--Muller codes. Classical Reed--Muller codes are constructed by evaluating multivariate polynomials at elements of an extension of $\Fq$. They are sometimes called affine Reed--Muller codes, in contraposition with projective Reed--Muller codes \cite{Lachaud88,Sor91}, where one evaluates homogenous polynomials over the elements of a projective space. Both families of Reed--Muller codes are fairly well-studied (see \cite{KLP68,DGMW70} for affine Reed--Muller codes and \cite{Lachaud90,GL23} for projective ones). In particular, computing the dimension of such codes boils down to computing the dimension of some space of polynomials of bounded degree, while for studying the minimum distance one needs to control the number of zeroes of multivariate polynomials. In the affine case this is a classical result \cite[Theorem 6.13]{LNbook}, while in the projective space the answer was given by Serre who proved a conjecture of Tsfasman \cite{Serre89}.

Coming back to the sum-rank metric, it is natural for constructing the analogue of Reed--Muller codes to look into multivariate Ore polynomials of bounded total degree. Therefore, firstly, we develop the theory of multivariate Ore polynomials and their evaluation. 
Secondly, we exploit this theory to propose the analogue of Reed--Muller codes in the sum-rank metric and study their parameters. Similarly to the Hamming case, the dimension is easily given by counting the number of monomials of a fixed degree, while to study the minimum distance we need to control the sum of the dimensions of the kernels of evaluations of multivariate Ore polynomials. This bound is proved in Theorem \ref{th:boundzeroesOre}; besides its application to linearized Reed--Muller codes, we believe it is interesting in itself. After giving the parameters of linearized Reed--Muller codes (Theorem \ref{th:codeparameters}), we prove in addition that by allowing some flexibility in the construction, we can obtain a larger panel of codes with better parameters estimations (Theorem \ref{th:codeparameters2}). Among all those codes, we get the best parameters when considering the ``almost commutative" case which corresponds to the Ore polynomial algebra $ \LL[X_1, \ldots, X_{\numvar-1}][X_\numvar; \Phi]$, where the only noncommutative variable is the last one. Our last contribution is to show in Theorem \ref{th:embedLAG} that, in many cases, linearized Reed--Muller codes embed in some linearized Algebraic Geometry (LAG) codes as introduced in \cite{BC24}. This could turn to be crucial  to decode the newly introduced linearized Reed--Muller codes as soon as a decoding algorithm for LAG codes will be available.

\smallskip
Finally, let us point out that several attempts have been made to generalise the construction of Reed--Muller codes using multivariate Ore polynomials. In \cite{GU19}, the authors introduced a notion of evaluation of multivariate Ore polynomials making use of noncommutative Gröbner basis, and defined skew Reed--Muller codes in the Hamming metric. However, the full study of the parameters of these codes is not provided. Another analogue of Reed--Muller codes, this time in the \emph{rank} metric, was introduced in \cite{ACLN21}. The first evident difference with the present paper is that we consider the sum-rank metric. More interestingly, we work over a finite field (for practical applications to coding theory) whereas the main motivation in  \cite{ACLN21} is to provide constructions in general abelian extensions. The latter falls in the finite fields setting only when the considered extension is cyclic, which was not the main case of interest in \cite{ACLN21}.  
\subsection*{Organisation of the paper}
Section~\ref{sec:ore} is devoted to the theory of rings of multivariate Ore polynomials and their evaluation. In Section~\ref{sec:thecode}, we introduce linearized Reed--Muller codes by evaluating multivariate Ore polynomials, and we study their parameters. Here we also show how one can improve on the parameters, and provide an example of our construction. Finally, in Section~\ref{sec:LAGcodes}, we outline the construction of linearized Algebraic Geometry codes introduced in \cite{BC24}, and prove that, in many cases, our new linearized Reed--Muller codes can be embedded in some LAG codes.

\section{Multivariate Ore polynomial rings}\label{sec:ore}
The algebra of univariate Ore polynomials was introduced by Ore in 1933 \cite{Ore33}. Its theory has been extensively studied, and exploited in algebraic and geometric rank and sum-rank metric codes. In this section, we partially develop the theory of multivariate Ore polynomials. 
Although similar questions were addressed before \cite{GU19,MP22}, they were not suitable for our purposes. On the one hand, the notion of evaluation developed in \cite{GU19} naturally leads to scalar values, hence to codes in the Hamming metric. On the other hand, the approach in \cite{MP22} is very general (as it does not require the variables to commute) and has not resulted in a construction of codes with controlled parameters so far.

We  refer the reader to \cite{Reiner75} for the theory of central simple algebras, and to \cite[Chapter III]{Lang3} and \cite{Eisenbud13} for classical results on modules and more in general in commutative algebra, which are used without reference in what follows.
\smallskip

Throughout the article, we let $\F$ be a finite field with $q$ elements, and $\LL$ be an extension of degree $r>0$. We consider $\Phi: \LL\to\LL$ to be the $q$-Frobenius endomorphism $x\mapsto x^q$.
For $\ue=(e_1,\dots,e_\numvar)\in\Z^\numvar$ such that $\mathrm{gcd}(e_1,\dots,e_\numvar,r)=1$, we consider the ring $\LL[X_1,\dots,X_\numvar; \Phi^{e_1},\dots,\Phi^{e_\numvar}]$ of multivariate Ore polynomials with usual sum, and multiplication given by 
\begin{align*}
&X_i\cdot X_j = X_j\cdot X_i,\\
&X_i \cdot a=\Phi^{e_i}(a) \cdot X_i, \quad\forall a\in\LL.
\end{align*}
We remark that this actually defines a ring since the $\Phi^{e_i}$ pairwise commute. For simplicity, we set $\theta_i=\Phi^{e_i}$, and write 
$\LL[\uX;\utheta]$ for $\LL[X_1,\dots,X_\numvar; 
\Phi^{e_1},\dots,\Phi^{e_\numvar}]$.

In what follows, we will often need to invert the variables 
$X_i$; this is possible because $\Phi$ is invertible, so that we can
extend the commutative relations to $X_i^{-1}$ by setting
$$X_i^{-1} \cdot a=\Phi^{-e_i}(a) \cdot X_i^{-1}, \quad\forall a\in\LL.$$
The resulting ring is denoted by 
$\LL[X_1^{\pm 1},\dots,X_\numvar^{\pm 1}; \Phi^{e_1},\dots,\Phi^{e_\numvar}]$, 
which we abbreviate as $\LL[\uX^{\pm 1};\utheta]$. 
For $\uu = (u_1, \ldots, u_\numvar) \in \Z^\numvar$, we also use the short
notation $\uX^{\uu}$ for the monomial $X_1^{u_1} \cdots X_\numvar^{u_\numvar}
\in \LL[\uX^{\pm 1};\utheta]$.
It is an easy computation to check that the general commutation
relation between a monomial and a scalar reads
$$\uX^{\uu} \cdot a = \Phi^{\ue\cdot\uu}(a) \cdot \uX^{\uu}, \quad
\forall \uu \in \Z^\numvar,\, a \in \LL$$
where, by definition, $\ue{\cdot}\uu=e_1u_1+\dots+e_\numvar u_\numvar$ is the
scalar product of $\ue$ and $\uu$.

\subsection{Evaluation of multivariate Ore polynomials}
\label{subsec:evalOre}

In the classical case, evaluation of polynomials are defined by giving
some value to the indeterminate.
In the Ore setting, we will not substitute scalar values but matrices
(or equivalently, linear maps); this is a crucial difference which
allows somehow to ``keep track'' of the non-commutativity.

It turns out nevertheless that Ore evaluation (or, more generally,
noncommutative evaluation) meets classical evaluation when one 
restricts to the centre.
Recall that the centre of a noncommutative ring $A$ is, by definition,
the subset $Z$ of $A$ consisting of elements $z \in A$ such that
$az = za$ for all $a \in A$. It is a commutative subring of $A$.
Besides, any ``matrix'' evaluation morphism $\varepsilon : A 
\to M_n(F)$ (for a certain field $F$) induces a ring homomorphism 
from $Z$ to the centre of $M_n(F)$, which is $F$. We then get an
evaluation morphism in the classical sense, taking values in a field.

Before defining evaluation of multivariate Ore polynomials, it is
therefore important to determine the centre of $\LL[\uX^{\pm 1};
\utheta]$.
For this, we introduce the following lattice:
\[L=\big\{\,\uu=(u_1,\dots,u_\numvar)\in\Z^\numvar \,|\, \ue{\cdot}\uu\in r\Z\,\big\}.\]
Note that $L$ is the kernel of the morphism $\Z^\numvar\to\Z/r\Z$, $\uu\mapsto \ue{\cdot}\uu$ which is surjective since $\mathrm{gcd}(e_1, \ldots, e_\numvar,r)=1$. So we get canonical isomorphisms $\Z^\numvar/L\simeq \Z/r\Z\simeq \Gal(\LL/\Fq)$. We define 
$$\F[\uX^{L}]=\left\{\,\sum_{\uu\in L} a_{\uu} X^{\uu} \text{ (finite sum)}\,|\, a_{\uu}\in\F\,\right\}.$$

\begin{proposition}
The centre of $\LL[\uX^{\pm 1};\utheta]$ is $\F[\uX^{L}]$.
\end{proposition}

\begin{proof}
The proof is analogue to the one for univariate Ore polynomials. Let $P=\sum_{\uu\in\Z^\numvar} a_{\uu} \uX^{\uu}$ be an element in the centre. For $i\in\{1,\dots,\numvar\}$, let $\ub_i$ be the $i$-th vector of the standard basis. Then
\[0=P\cdot X_i- X_i \cdot P = \sum_{\uu\in\Z^\numvar} (a_{\uu}-\theta_i(a_{\uu})) \uX^{\uu+\ub_i}.\]
Therefore, we must have $a_{\uu}=\theta_i(a_{\uu})$ for any $i$. Since $\theta_i=\Phi^{e_i}$ and $\mathrm{gcd}(e_1, \ldots, e_\numvar,r)=1$, we entail $a_{\uu}\in\F$. Now, take $a\in\LL$. Then
\[0=P\cdot a- a\cdot P = \sum_{\uu\in\Z^\numvar} a_{\uu}(\Phi^{\ue\cdot\uu}(a)-a) \uX^{\uu}.\]
Therefore, we must have $ \Phi^{\ue\cdot\uu}(a)=a$ for all $a\in \LL$. Hence $\ue{\cdot}\uu\in r\Z$, that is $\uu\in L$.
\end{proof}

Let us now consider a ring homomorphism
$\varepsilon : \F[\uX^{L}] \to \F$. It is of the form
\[\begin{array}{c r c l}
\epsilon_\gamma: & \F[\uX^L] & \rightarrow & \F \smallskip \\
& \sum_{\uu\in L} a_{\uu}\uX^{\uu} & \mapsto & \sum_{\uu\in L} a_{\uu} \gamma(\uu),
\end{array}\]
where $\gamma :L\to \F^\times$ is a group morphism. Our goal is to extend
$\epsilon_\gamma$ to a second ring homomorphism $\LL[\uX^{\pm 1};\utheta]
\to \End_{\F}(\LL) \simeq M_r(\F)$. We will search the latter among
morphisms of the form
\[\begin{array}{c r c l}
\epsilon_{\tilde \gamma}: 
& \LL[\uX^{\pm 1}; \utheta] & \rightarrow & \End_{\F}(\LL) \smallskip \\
& \sum_{\uu\in \Z^\numvar} a_{\uu}\uX^{\uu} & \mapsto  &
  \sum_{\uu\in \Z^\numvar} a_{\uu} \tilde \gamma(\uu) \Phi^{\ue\cdot\uu}
\end{array}\]
where $\tilde \gamma : \Z^\numvar \to \LL^\times$ is a function extending~$\gamma$.
One checks that $\epsilon_{\tilde \gamma}$ is a ring homomorphism if and
only if $\tilde \gamma$ satisfies the following property
\begin{equation}\label{cocycle}
\forall \uu, \uv \in \Z^\numvar, \quad
\tilde \gamma (\uu + \uv)= 
  \tilde \gamma (\uu) \cdot \Phi^{\ue \cdot\uu}(\tilde \gamma(\uv)).
\end{equation}

\begin{lemma}
\label{lem:cocycle}
The function $\gamma$ extends to a function $\tilde \gamma : \Z^\numvar \to \LL^\times$
satisfying the axiom~\eqref{cocycle}.
\end{lemma}

\begin{proof}
Since $\Z^\numvar/L$ is isomorphic to $\Z/r\Z$, the theorem of structure of
$\Z$-modules ensures that there exists
a basis $(\uv_1,\dots,\uv_\numvar)$ of $\Z^\numvar$ such that  $(\uv_1,\dots,\uv_{\numvar-1},r\uv_\numvar)$ is a basis of $L$.
We can moreover assume that $\ue{\cdot}\uv_\numvar \equiv 1 \pmod r$.
Let $\alpha \in \LL^\times$ be a preimage of $\gamma(r \uv_\numvar) \in \F^\times$ 
by the norm map $N_{\LL/\F}$, \ie
$$\gamma(r \uv_\numvar) = N_{\LL/\F}(\alpha) = 
\alpha \cdot \Phi(\alpha) \cdots \Phi^{r-1}(\alpha).$$
For $a_1, \ldots, a_\numvar \in \Z$, we write $a_\numvar = q_\numvar r + r_\numvar$ with $1 \leq
r_\numvar \leq r$ and set
$$\tilde \gamma\big(a_1 \uv_1 + \cdots + a_\numvar \uv_\numvar)
= \gamma(\uv_1)^{a_1} \cdots \gamma(\uv_{\numvar-1})^{a_{\numvar-1}} \cdot \gamma(r \uv_\numvar)^{q_\numvar}
\cdot \alpha \cdot \Phi(\alpha) \cdots \Phi^{r_\numvar-1}(\alpha).$$
One finally checks that $\tilde \gamma$ satisfies the requirements of the lemma.
\end{proof}

\begin{remark}
The preimage $\alpha$ is not unique, implying that there are in 
general
many $\tilde \gamma$ extending $\gamma$. However, two appropriate $\alpha$ 
always differ by multiplication by an element of norm $1$, which 
eventually ensures that the morphisms $\epsilon_{\tilde \gamma}$ we get 
at the end of the process are conjugated.
\end{remark}

\begin{remark}
A function $\tilde \gamma$ respecting property \eqref{cocycle} 
is called a \emph{$1$-cocycle}. In fact, Lemma~\ref{lem:cocycle} can 
also be obtained as a consequence of the inflation-restriction exact 
sequence in group cohomology, that in our context reads
\[0\to H^1(\Z^\numvar/L,\,\LL^\times)\to H^1(\Z^\numvar,\,\LL^\times) \to 
\Homgrp(L,\,\F^\times) \to H^2(\Z^\numvar/L,\,\LL^\times).\]
Noting that $\Z^\numvar/L \simeq \Gal(\LL/\F)$, we find that
$H^1(\Z^\numvar/L,\,\LL^\times)$ is trivial by Hilbert 90 Theorem and 
that $H^2(\Z^\numvar/L,\,\LL)$ is the Brauer group of  $\LL/\F$ which is 
trivial too, given that $\F$ is a finite field.
Therefore, we get an isomorphism $H^1(\Z^\numvar,\,\LL^\times) \simeq
\text{Hom}_{\text{grp}}(L,\,\F^\times)$, which means that
any $\gamma\in \text{Hom}_{\text{grp}}(L,\,\F^\times)$ extends to a 
$1$-cocycle  $\tilde \gamma\in H^1(\Z^\numvar,\,\LL^\times)$ which is unique 
up to a $1$-coboundary. We refer the reader to \cite{Serre79} for more details on group cohomology. \end{remark}

\begin{theorem}
\label{th:evalisom}
We keep the above notation.
Let $\mathfrak{m}_\gamma$ be the ideal of $\F[\uX^L]$ generated by 
the $\uX^{\uu} - \gamma(\uu)$ for $\uu\in L$, 
\ie $\m_\gamma = \ker \epsilon_\gamma$.
Then $\epsilon_{\tilde \gamma}$ induces an isomorphism
\[\begin{array}{r c l}
\LL[\uX^{\pm 1};\utheta] /\mathfrak{m}_\gamma \LL[\uX^{\pm 1};\utheta] 
& \stackrel\sim\longrightarrow & \End_{\F}(\LL) \smallskip \\
\sum_{\uu\in \Z^\numvar} a_{\uu}\uX^{\uu} & \mapsto  &
  \sum_{\uu\in \Z^\numvar} a_{\uu} \tilde \gamma(\uu) \Phi^{\ue\cdot\uu}.
\end{array}\]
\end{theorem}
\begin{proof}
By Artin's theorem on independence of characters the family $\{\text{Id},\Phi,\dots\Phi^{r-1}\}$ generates $\End_{\F}(\LL)$ over $\LL$, whence the surjectivity. Injectivity follows by comparing the dimensions over $\F$.
\end{proof}

\subsection{Reduced norm}

A fundamental tool that makes the connection between Ore polynomials and 
classical polynomials is the reduced norm. Concretely, it takes the form 
of a multiplicative map from $\LL[\uX^{\pm 1};\utheta]$ to its centre 
$\F[\uX^L]$.

In order to define it, it is convenient to introduce intermediate
rings between $\F[\uX^L]$ and $\LL[\uX^{\pm 1};\utheta]$.
In what follows, we shall consider two of them, namely $\CC_1 = 
\F[\uX^{\pm 1}]$ and $\CC_2 = \LL[\uX^L]$. We observe that both of
them are commutative and endow $\LL[\uX^{\pm 1};\utheta]$ with a
structure of $\CC_i$-module ($i = 1, 2$): for $c \in \CC_i$ and 
$f \in \LL[\uX^{\pm 1};\utheta]$, the outcome of the action of $c$ on
$f$ is simply the product $cf$ computed in $\LL[\uX^{\pm 1};\utheta]$.
We note that $\LL[\uX^{\pm 1};\utheta]$ is free of rank $r$ over 
$\CC_1$ and $\CC_2$. In the former case, a basis is given by a
basis of $\LL$ over $\F$ while, in the latter, it is formed 
by the $\uX^{\uu}$ where $\uu$ runs over a set of representatives 
of $\Z^\numvar/L$.

For $i = 1, 2$, we define the norm map 
$N_i : \LL[\uX^{\pm 1};\utheta] \to \CC_i$ as follows. Given a Ore
polynomial $f \in \LL[\uX^{\pm 1};\utheta]$, we consider the map
\[\begin{array}{c r c l}
\mu_f : & \LL[\uX^{\pm 1};\utheta] & \rightarrow & \LL[\uX^{\pm 1};\utheta] \\
& g & \mapsto & gf
\end{array}\]
and view it as a $\CC_i$-linear endomorphism. We then set
$N_i(f) = \det_{\CC_i} (\mu_f)$. We note that, working in the bases we 
have mentioned earlier, it is possible to write down explicitly the 
matrix of $\mu_f$. This provides an efficient method for computing the 
maps $N_1$ and $N_2$.

\begin{theorem}
\label{th:Nrd}
For all $f \in \LL[\uX^{\pm 1};\utheta]$, we have $N_1(f) = N_2(f)
\in \F[\uX^L]$.
\end{theorem}

A key step in the proof of Theorem~\ref{th:Nrd} is the following
proposition.

\begin{proposition}
\label{prop:azumaya}
For $i = 1, 2$, the map
\[\begin{array}{c r c l}
\iota_i :
& \CC_i \otimes_{\F[\uX^L]} \LL[\uX^{\pm 1};\utheta] & \longrightarrow 
& \End_{\CC_i}\big(\LL[\uX^{\pm 1};\utheta]\big) \simeq M_r(\CC_i) \\
& c \otimes f & \mapsto & c \cdot \mu_f
\end{array}\]
is an isomorphism of $\CC_i$-algebras.
\end{proposition}

\begin{proof}
It is routine to check that $\iota_i$ is a morphism of 
$\CC_i$-algebras. Since the domain and the codomain are both
free of rank $r^2$ over $\CC_i$, it is enough to show that 
$\iota_i$ is surjective.

We start with $\iota_1$. Its image contains obviously the 
multiplication by the elements of $\LL$. Besides, for $\uu \in 
\Z^\numvar$, we notice that $\iota_1$ takes the element $\uX^{\uu} 
\otimes \uX^{-\uu}$ to $\Phi^{\ue\cdot\uu}$ (acting coefficient-wise 
on the Ore polynomial). By Artin's theorem, we conclude that the 
image of $\iota_1$ contains at least $\End_{\F}(\LL) \simeq M_r(\F)$
(see also the proof of Theorem~\ref{th:evalisom}).
Since it is in addition a $\CC_1$-module, we conclude that
$\text{im } \iota_1 = M_r(\CC_1)$ and we are done.

We now move to $\iota_2$. Let $\uv \in \Z^\numvar$ be an element such
that $\ue {\cdot}\uv \equiv 1 \pmod r$.
Then $\big(1, \uX^{\uv}, \uX^{2\uv}, \ldots, \uX^{(r-1)\uv}\big)$ is
a basis of $\LL[\uX^{\pm 1};\utheta]$ over $\CC_2$ and we use it to
identify $\End_{\CC_2}\big(\LL[\uX^{\pm 1};\utheta]\big)$ with
$M_r(\CC_2)$. One checks
that, for any $\lambda, \alpha \in \LL$, 
\[\iota_2(\lambda \otimes \alpha) =
\left(\begin{matrix}
\lambda \alpha \\
& \lambda \Phi(\alpha) \\
& & \ddots \\
& & & \lambda \Phi^{r-1}(\alpha)
\end{matrix}\right).\]
Noticing that the map $\LL \otimes_{\F} \LL \rightarrow \LL^r$,
$\lambda \otimes \alpha \mapsto 
  \big(\lambda \Phi^i(\alpha)\big)_{0 \leq i < r}$
is an isomorphism by Galois theory, we find that the image of
$\iota_2$ contains all diagonal matrices with coefficients 
in $\LL$. Therefore it contains more generally all diagonal matrices 
with coefficients in $\CC_2$ given that it is a module over $\CC_2$.
Finally, we observe that
\[ \iota_2(1 \otimes \uX^{\uv}) =
\left(\begin{matrix}
& 1 \\
& & \ddots \\
& & & 1 \\
\uX^{r\uv} \\
\end{matrix}\right).\]
Since the latter together with diagonal matrices generate 
$M_r(C_2)$, we conclude that $\iota_2$ is surjective.
\end{proof}

\begin{proof}[Proof of Theorem~\ref{th:Nrd}]
We define $\CC = \CC_1 \otimes_{\F[\uX^L]} \CC_2 \simeq
\LL[\uX^{\pm 1}]$, so that we have the following isomorphisms
of $\CC$-algebras:
\begin{align*}
\CC \otimes_{\F[\uX^L]} \LL[\uX^{\pm 1};\utheta] 
& \simeq \CC_1 \otimes_{\F[\uX^L]} M_r(\CC_2) \simeq M_r(\CC)
\qquad \text{(\emph{via} $\CC_1 \otimes \iota_2$)} \\
& \simeq \CC_2 \otimes_{\F[\uX^L]} M_r(\CC_1) \simeq M_r(\CC)
\qquad \text{(\emph{via} $\CC_2 \otimes \iota_1$)}.
\end{align*}
It follows from the Skolem--Noether theorem \cite[Theorem 7.21]{Reiner75}
that the two above 
isomorphisms are conjugated over $\text{Frac }\CC$. In other words, 
there exists a matrix $P \in \text{GL}_r(\text{Frac }\CC)$ with the 
property that, for all $f \in \LL[\uX^{\pm 1};\utheta]$, one has
$$\text{Mat}_1(\mu_f) = P^{-1} \cdot \text{Mat}_2(\mu_f) \cdot P,$$
where $\text{Mat}_i(\mu_f)$ denotes the matrix of $\mu_f$ when it is
viewed as a $\CC_i$-linear map. Taking determinants, we end up with
$N_1(f) = N_2(f)$.

Finally, given that $N_1$ and $N_2$ take values respectively in
$\CC_1 = \F[\uX^{\pm 1}]$ and $\CC_2 = \LL[\uX^L]$, the equality 
of those maps implies that they must assume values in the intersection
$\CC_1 \cap \CC_2$, which is $\F[\uX^L]$.
\end{proof}

The map $N_1 = N_2$ is called the \emph{reduced norm map} and it
will be denoted by $\Nrd$ in what follows. We will always consider
it as a map from $\LL[\uX^{\pm 1};\utheta]$ to $\F[\uX^L]$.
We record the following facts which immediately follow from the 
corresponding properties of the determinant:
\begin{enumerate}[(i)]
\item the map $\Nrd$ is multiplicative, \ie 
for all $f, g \in \LL[\uX^{\pm 1};\utheta]$,
$$\Nrd(fg) = \Nrd(gf) = \Nrd(f) \cdot \Nrd(g)$$
\item for $f \in \F[\uX^L]$, we have $\Nrd(f) = f^r$.
\end{enumerate}

\medskip

Coming back to the definition, one can also easily obtain bounds on
the size of the reduced norm of a given Ore polynomial. For example,
we have the following.

\begin{lemma}
\label{lem:nrddegree}
Let $f \in \LL[\uX;\utheta]$ be of total degree $\degpol$.
Then $\Nrd(f) \in \F[\uX^L] \cap \F[\uX]$ and its total degree
with respect to the variables $\uX$ is at most $r\degpol$.
\end{lemma}

\begin{proof}
We use the description coming from the subring $\CC_1$.
Let $\BB = (\alpha_1, \ldots, \alpha_r)$ be a basis of $\LL$ over $\F$.
The entries of the matrix $\text{Mat}_{\BB}(\mu_f)$ representing
$\mu_f$ in the basis $\BB$ gathers the coordinates of the 
$f \alpha_i$ ($1 \leq i \leq r$) in the basis $\BB$. Therefore, 
they are all
polynomials in $\F[\uX]$ of degree at most $\degpol$. The determinant
of $\text{Mat}_{\BB}(\mu_f)$, which is also $\Nrd(f)$, is then
a polynomial in $\F[\uX]$ of degree at most $r\degpol$.
\end{proof}

Another decisive property of the reduced norm map is that its 
vanishing controls the size of the kernels of the evaluation morphisms
$\epsilon_{\tilde \gamma}$ introduced earlier.
In order to state a precise result in this direction, we need to
introduce the \emph{order of vanishing} of a central function: given
$f \in \F[\uX^L]$ and a group morphism $\gamma : L \to \F^\times$, we
define
$$\ord_\gamma(f) = 
  \inf \big\{\, v \in \N \,|\, f \in \m_\gamma^v \,\big\}$$
where we recall that $\m_\gamma = \ker \epsilon_\gamma$ is the
ideal of $\F[\uX^L]$ generated by the $\uX^{\uu} - \gamma(\uu)$, 
$\uu \in L$. By convention $\ord_\gamma(0) = +\infty$ for all
$\gamma$.

\begin{theorem}
\label{th:dimker}
Let $f \in \LL[\uX^{\pm 1};\utheta]$. Let also $\gamma : L \to
\F^\times$ be a group morphism and $\tilde \gamma: \Z^\numvar
\to \LL^\times$ be a prolongation of $\gamma$ satisfying the 
cocycle condition~\eqref{cocycle}. Then
$$\dim_{\F} \ker \epsilon_{\tilde \gamma}(f)
  \leq \ord_\gamma\big(\Nrd(f)\big).$$
\end{theorem}

\begin{proof}
Throughout the proof, we write $\CC = \CC_2 = \LL[\uX^L]$.
As in the proof of Lemma~\ref{lem:nrddegree}, let
$\BB = (\alpha_1, \ldots, \alpha_r)$ be a basis of $\LL$ over $\F$.
We also fix a basis of $\LL[\uX^{\pm 1};\utheta]$ over $\CC$
and write $\text{Mat}(\mu_f)$ for the matrix of $\mu_f$ in this
basis. By definition $\Nrd(f) = \det \text{Mat}_{\BB}(\mu_f)$.
Besides, we notice that $\epsilon_\gamma$ induces an 
isomorphism between $\F[\uX^L]/\m_\gamma$ and $\F$. Therefore it 
also induces an isomorphism $\CC/\m_\gamma\CC \simeq \LL$. It
then follows from Proposition~\ref{prop:azumaya} that we have an
isomorphism
\[\begin{array}{r c l}
\CC/\m_\gamma\CC \otimes_{\F[\uX^L]} \LL[\uX^{\pm 1};\utheta]
& \stackrel\sim\longrightarrow & M_r(\LL) \\
\lambda \otimes f & \mapsto & 
 \lambda \cdot \text{Mat}(\mu_f) \text{ mod }\m_\gamma\CC.
\end{array}\]
On the other hand, it follows from Theorem~\ref{th:evalisom} that
the evaluation morphism $\epsilon_{\tilde \gamma}$
induces another isomorphism
\[\begin{array}{r c l}
\CC/\m_\gamma\CC \otimes_{\F[\uX^L]} \LL[\uX^{\pm 1};\utheta]
& \stackrel\sim\longrightarrow & M_r(\LL)
\end{array}\]
after scalar extension to $\LL$. By Skolem--Noether theorem, those two 
isomorphisms are conjugated: there exists a matrix $P \in \text{GL}_r
(\LL)$ such that, for all $f \in \LL[\uX^{\pm 1};\utheta]$,
\begin{equation}
\label{eq:conjugated}
\text{Mat}_{\BB}\big(\epsilon_{\tilde \gamma}(f)\big) 
  \equiv P^{-1} \cdot \text{Mat}(\mu_f) \cdot P
  \pmod{\m_\gamma\CC}.
\end{equation}
Write $\delta= \dim_{\F} \ker \epsilon_{\tilde \gamma}(f)$ and 
pick a basis $\BB' = (\alpha'_1, \ldots, \alpha'_r)$ of $\LL$
over $\F$ such that $\alpha'_1, \ldots, \alpha'_\delta$ generate
$\ker \epsilon_{\tilde \gamma}(f)$.
Let $Q \in \text{GL}_r(\F)$ be the 
change-of-basis matrix between $\BB$ and~$\BB'$.
Equation~\eqref{eq:conjugated} then gives
$$
\text{Mat}_{\BB'}\big(\epsilon_{\tilde \gamma}(f)\big) 
  \equiv (PQ)^{-1} \cdot \text{Mat}(\mu_f) \cdot PQ
  \pmod{\m_\gamma\CC}.$$
Therefore the matrix of $\mu_f$ is conjugated to a matrix whose 
$\delta$ first rows vanish modulo $\m_\gamma\CC$. As a consequence, 
its determinant $\Nrd(f)$ falls inside $\m_\gamma^\delta \CC$.
Given that $\Nrd(f)$ also lies in $\F[\uX^L]$, we find
$\Nrd(f) \in \m_\gamma^\delta$ which is equivalent to say that 
$\ord_\gamma\big(\Nrd(f)\big) \geq \delta$.
\end{proof}

\section{Linearized Reed--Muller codes}\label{sec:thecode}

In this section we introduce codes in the sum-rank metric constructed 
by evaluating multivariate Ore polynomials, that we call linearized 
Reed--Muller codes.

\subsection{The code construction}\label{ssec:construction}

We keep the notation of Section \ref{sec:ore}. Briefly, we recall
that $\Phi$ denotes the Frobenius endomorphism $x \mapsto x^q$ acting
on $\LL$. We pick a tuple $\ue = (e_1, \ldots, e_\numvar) \in \Z^\numvar$ such
that $\gcd(e_1, \ldots, e_\numvar, r) = 1$. 
We set $\utheta = (\Phi^{e_1}, \cdots, \Phi^{e_\numvar})$ and consider
the ring of Ore polynomials $\LL[\uX^{\pm 1}; \utheta]$. We recall
that its centre is $\F[\uX^L]$ where $L$ is the kernel of the map
$\Z^\numvar \to \Z/r\Z$, $\uu \mapsto \ue{\cdot}\uu$. It is a lattice in
$\Z^\numvar$ satisfying $\Z^\numvar/L \simeq \Z/r\Z \simeq \Gal(\LL/\F)$.

Let $H := \Homgrp(L, \F^\times)$ be the set of group
homomorphisms from $L$ to $\F^\times$. We remark that giving an element of $H$ is the same as giving its values on a fixed basis of $L$ (over $\Z$), \ie $m$ elements of $\F^\times$.  Therefore, $ \Card(H)=(q-1)^m$.

For each $\gamma \in H$, we fix
a prolongation $\tilde \gamma : \Z^\numvar \to \LL^\times$ of $\gamma$ 
satisfying the cocycle condition~\eqref{cocycle}. It follows from 
Lemma~\ref{lem:cocycle} that such a prolongation always exists.
Besides, we recall that it gives rise to an evaluation morphism 
$\epsilon_{\tilde \gamma}: 
\LL[\uX^{\pm 1}; \utheta] \to \End_{\F}(\LL)$ (see Subsection
\ref{subsec:evalOre}).
We put together all of those into a unique multievaluation 
morphism
\[\begin{array}{c r c l}
\epsilon: & \LL[\uX^{\pm 1}; \utheta] & \rightarrow & 
\prod_{\gamma \in H} \End_{\F}(\LL)
 \smallskip \\
& f & \mapsto & \big(\epsilon_{\tilde \gamma}(f)\big)_{\gamma \in H}
\end{array}.\]
The codomain of $\epsilon$, namely
\[\HH := 
\prod_{\gamma \in H} \End_{\F}(\LL),\]
will play an important role in what follows: it is
the space in which all our codes will eventually sit.
It is a vector space over $\LL$; indeed given a scalar $a \in \LL$
and a $\F$-linear endomorphism $f$ of $\LL$, the product $af$ 
makes sense: it is simply the map $\LL \to \LL$ that takes $x \in
\LL$ to $a{\cdot}f(x)$. We underline that, for this structure, the
map $\epsilon$ is $\LL$-linear. Besides we notice that $\End_{\F}(\LL)$ 
has dimension $r^2$ over $\F$, and hence it has dimension $r$ over
$\LL$. Therefore
$$\dim_{\LL} \HH = r \cdot \Card(H) = r \cdot (q{-}1)^\numvar.$$
Moreover, $\HH$ is endowed with the sum-rank metric. Precisely, 
we define the \emph{sum-rank weight} of a tuple 
$\boldsymbol{\varphi} = (\varphi_\gamma)_{\gamma \in H} \in \HH$ 
by
\[
  w_\srk(\boldsymbol{\varphi}) = \sum_{\gamma \in H} \rk(\varphi_\gamma).
\]

\begin{definition}
\label{def:codesLRM}
Let $\ue$ be as above and $\degpol$ be a positive integer.
The \emph{linearized Reed-Muller code} associated to $\ue$ and 
$\degpol$ is
$$\LRM(\ue;\,\degpol) = \epsilon\Big(\LL[\uX^{\pm 1}; \utheta]_{\leq \degpol}\Big)$$
where $\LL[\uX^{\pm 1}; \utheta]_{\leq \degpol}$ is the subspace of
$\LL[\uX^{\pm 1}; \utheta]$ consisting of multivariate Ore polynomials of
total degree at most $\degpol$.
\end{definition}

Since $\epsilon$ is $\LL$-linear, the code $\LRM(\ue;\,\degpol)$ is
$\LL$-linear as well, \ie it is a $\LL$-vector subspace
of $\HH$.

\begin{remark}
When $e_1=\ldots=e_m=r=1$, the ring $\LL[\uX^{\pm 1}; \utheta]$ is the classical ring of Laurent polynomials in $m$ variables, and the code $\LRM(\ue;\,\degpol)$ is the classical Reed--Muller code, except that we do not allow evaluation at tuples with a zero component.
\end{remark}

\subsection{Code's parameters}\label{ssec:parameters}

We recall from the introduction that, for a $\LL$-linear code $\CC$ 
sitting inside $\HH$, we define:
\begin{itemize}
\item its \emph{length} $n$ as the $\LL$-dimension of the ambient 
space $\HH$, \ie $n \coloneqq r \cdot (q{-}1)^\numvar$,
\item its \emph{dimension} $k$ as its $\LL$-dimension, \ie 
$k \coloneqq \dim_{\LL} \CC$,
\item its \emph{minimum distance} $\mindist$ as the minimum sum-rank
weight of a nonzero codeword in~$\CC$.
\end{itemize}
The next theorem provides the dimension and an explicit lower bound for the minimum distance of our codes.

\begin{theorem}\label{th:codeparameters}
Let $\ue = (e_1, \ldots, e_\numvar) \in \Z^\numvar$ with $\gcd(e_1, \ldots, 
e_\numvar, r) = 1$ as before.
Let also $\degpol$ be an integer between $1$ and $q{-}2$.
Then, the dimension $k$ and the minimum distance $\mindist$ of
$\LRM(\ue;\,\degpol)$ satisfy
\[k = \binom{\degpol+\numvar}{\degpol} 
\qquad \text{and} \qquad
\mindist \geq r \cdot (q{-}1)^{\numvar-1} \cdot (q - 1 - \degpol).\]
\end{theorem}

The rest of this subsection is devoted to the proof of
Theorem~\ref{th:codeparameters}. A crucial ingredient is an
upper bound on the ``number of zeroes'' of a multivariate Ore
polynomial in the spirit of the Schwartz--Zippel lemma. Before addressing this question, we recall the
corresponding result for classical multivariate polynomials.

\begin{proposition}
\label{prop:boundzeroes}
Let $f \in \F[\uX]$ be a nonzero polynomial of total degree
at most $\degpol$. Then
\[
  \sum_{\ua \in (\F^\times)^\numvar} \ord_{\ua}(f) \leq \degpol \cdot (q{-}1)^{\numvar-1}
\]
where $\ord_{\ua}$ denotes the order of vanishing at $\ua = (a_1,
\ldots, a_\numvar)$, that is,
by definition, the smallest integer $v$ such that $f \in \m_{\ua}^v$
where $\m_{\ua}$ is the ideal generated by $X_1 - a_1, \ldots, X_\numvar
- a_\numvar$.
\end{proposition}

\begin{proof} See \cite[Lemma 2.7]{DKSS13}.
\end{proof}

We now move to the case of multivariate Ore polynomials, for which we have a
direct analogue of Proposition~\ref{prop:boundzeroes}.

\begin{theorem}
\label{th:boundzeroesOre}
Let $f \in \LL[\uX; \utheta]$ be a nonzero Ore polynomial of total
degree at most $\degpol$. Then
\[
  \sum_{\gamma \in H} \dim_{\F} \ker \epsilon_{\tilde \gamma}(f)
\leq r\degpol \cdot (q{-}1)^{\numvar-1}.
\]
\end{theorem}

\begin{proof}
The basic idea of the proof is to use Theorem~\ref{th:dimker}
and to apply Proposition~\ref{prop:boundzeroes} to the reduced
norm of~$f$; however, this requires some precaution.
On the one hand, we know from Lemma~\ref{lem:nrddegree} that the
reduced norm of $f$ is a polynomial in $\F[\uX^L]$ of total degree 
at most $r\degpol$. On the other hand,
we need to be careful before applying Proposition~\ref{prop:boundzeroes} 
because the evaluation points we are interested in correspond to group 
homomorphisms $L \to \F^\times$, which are not exactly the classical 
ones (which rather correspond to morphisms $\Z^\numvar \to \F^\times$).

In order to relate them, two ingredients are needed.
First of all, we need to compare the orders of vanishing for
the two types of evaluation points we are dealing with. 
Let then $g \in \Fq[\uX^L]$ be a central function. Of course,
thanks to the inclusion $\Fq[\uX^L] \subset \Fq[\uX^{\pm 1}]$,
$g$ can also be considered as a classical Laurent polynomial.
Let $\ua \in (\F^\times)^\numvar$ be an evaluation point. To it, we
attach the group morphism $\gamma' : \Z^\numvar \to \F^\times$,
$\uu \mapsto \ua^{\uu}$ where by definition 
$\ua^{\uu} = a_1^{u_1} \cdots a_\numvar^{u_\numvar}$ (with obvious notation).
We define $\gamma = \gamma'_{|L}$ as the restriction of $\gamma'$
to the lattice $L \subset \Z^\numvar$.
As in Proposition~\ref{prop:boundzeroes}, we consider
the ideal generated by $X_i - a_i$, $1 \leq i \leq \numvar$. However, 
for this proof, it will be more convenient to work over the ring
$\Fq[\uX^{\pm 1}]$ (instead of $\Fq[\uX]$). For this reason, we
define $\m_{\ua}$ by
\[
  \m_{\ua}
  = \left< X_1 - a_1, \cdots, X_\numvar - a_\numvar \right>_{\Fq[\uX^{\pm 1}]}
\]
where the notation means that we consider the generated ideal.
We underline that this modification does not change the order of
vanishing at $\ua$ since the coordinates of $\ua$ do not vanish by
assumption. However, it allows us to perform changes-of-basis.
Precisely, we consider a basis $(\uv_1, \ldots, \uv_\numvar)$ of $\Z^\numvar$
such that $(\uv_1, \ldots, \uv_{\numvar-1}, r\uv_\numvar)$ is a basis of $L$. 
We have
\[
  \m_{\ua}
  = \left< \uX^{\uv_1} - \gamma'(\uv_1),\, \ldots,\, 
           \uX^{\uv_\numvar} - \gamma'(\uv_\numvar) \right>_{\Fq[\uX^{\pm 1}]}.
\]
Similarly, we know that
\begin{align*}
  \m_{\gamma}
& = \left< \uX^{\uv_1} - \gamma(\uv_1),\, \ldots,\, \uX^{\uv_{\numvar-1}} - \gamma(\uv_{\numvar-1}),\,
           \uX^{r\uv_\numvar} - \gamma(r\uv_\numvar) \right>_{\Fq[\uX^L]} \\
& = \left< \uX^{\uv_1} - \gamma'(\uv_1),\, \ldots,\, \uX^{\uv_{\numvar-1}} - \gamma'(\uv_{\numvar-1}),\,
           \uX^{r\uv_\numvar} - \gamma'(\uv_\numvar)^r \right>_{\Fq[\uX^L]}.
\end{align*}
Observing that $\uX^{r\uv_\numvar} - \gamma'(\uv_\numvar)^r$ is a multiple of
$\uX^{\uv_\numvar} - \gamma'(\uv_\numvar)$ in $\Fq[\uX^{\pm 1}]$, we conclude
that $\m_{\gamma} \subset \m_{\ua}$. For all nonnegative integer
$v$, we thus have $\m_\gamma^v \subset \m_{\ua}^v$ as well, which
finally shows that
\begin{equation}
\label{eq:inegord}
\ord_\gamma(g) \leq \ord_{\ua}(g).
\end{equation}

The second important ingredient we shall need is a study of the 
prolongations to $\Z^\numvar$ of group morphisms $L \to \F^\times$. 
Continuing to work in our distinguished basis $(\uv_1, \ldots,
\uv_\numvar)$, we see that
a morphism $\gamma \in H$ extends to $\Z^\numvar$ if and only if 
$\gamma(r \uv_\numvar)$ is a $r$-th power in $\F^\times$; in particular,
it is not always the case. In order to handle
this difficulty, we introduce, for each $t \in \F^\times$, the 
endomorphism of $\F$-algebras $\sigma_t : \F[\uX^L] \to \F[\uX^L]$ 
defined by
$$\sigma_t: \quad 
  \uX^{\uv_i} \mapsto \uX^{\uv_i} \,\, (1 \leq i < \numvar), \quad
  \uX^{r\uv_\numvar} \mapsto t \uX^{r\uv_\numvar}$$
and similarly, we introduce the map $\rho_t : H \to H$ that takes
$\gamma$ to the group homomorphism $\rho_t(\gamma)$ defined by
$$\rho_t(\gamma): \quad 
  \uv_i \mapsto \gamma(\uv_i) \,\, (1 \leq i < \numvar), \quad
  r\uv_\numvar \mapsto t \cdot \gamma(r\uv_\numvar).$$
One easily checks that $\sigma_t$ is an isomorphism (with 
inverse $\sigma_{t^{-1}}$) and that it takes the maximal ideal 
$\m_{\gamma}$ to $\m_{\rho_t(\gamma)}$.

Let $R \subset \F^\times$ be a set of 
representatives of the quotient $\F^\times / (\F^\times)^r$
where $(\F^\times)^r$ denotes the subgroup of $\F^\times$ of $r$-th
powers. Since $\F^\times$ is cyclic of order $q{-}1$, $R$ has
cardinality $\gcd(r, q{-}1)$. We also consider the set $H'$ of
group homomorphisms $\Z^\numvar \to \F^\times$, together with the map
$\iota : R \times H' \to H$,
  $(t, \gamma') \mapsto \rho_t\big(\gamma'_{|L}\big)$.
We claim that the preimage of any $\gamma \in H$ under
$\iota$ has cardinality $\gcd(r, q{-}1)$.
Indeed, a pair $(t, \gamma')$ has image $\gamma$ if and only if
$\gamma'(\uv_i) = \gamma(\uv_i)$ for $i \in \{1, \ldots, \numvar{-}1\}$
and 
\begin{equation}
\label{eq:rthroot}
t \cdot \gamma'(\uv_\numvar)^r = \gamma(r \uv_\numvar).
\end{equation}
The latter condition is realized for exactly one element $t \in
R$, namely the representative of $\gamma(r \uv_\numvar)$. Besides, once
$t$ is known, the solutions of Equation~\eqref{eq:rthroot} are 
in (noncanonical) one-to-one correspondence with the group 
$\mu_r(\F)$ of $r$-th roots of unity in $\F$. Using again that
$\F^\times$ is cyclic of order $q{-}1$, we find that the
cardinality of $\mu_r(\F)$ is $\gcd(r, q{-}1)$, which proves our 
claim.

Summing over all pairs $(t, \gamma') \in R \times H'$ and
using the inequality~\eqref{eq:inegord}, we end up with
\begin{equation}
\label{eq:sumH}
\sum_{\gamma \in H} \ord_{\gamma}\big(\Nrd(f)\big)
= \frac 1{\gcd(r, q{-}1)} \sum_{t \in R} 
\sum_{\ua \in (\F^\times)^\numvar} \ord_{\ua}\big(\sigma_t\big(\Nrd(f)\big)\big).
\end{equation}
We recall from Lemma~\ref{lem:nrddegree} that
$\Nrd(f)$ has total degree at most $r\degpol$. It is then also the
case for all the $\sigma_t\big(\Nrd(f)\big)$ since applying $\sigma_t$
only affects the coefficients, leaving the exponents unchanged.
Therefore we can apply Proposition~\ref{prop:boundzeroes} to
those polynomials and obtain
\[
  \sum_{\ua \in (\F^\times)^\numvar} \ord_{\ua}\big(\sigma_t\big(\Nrd(f)\big)\big)
  \leq r\degpol \cdot (q{-}1)^{\numvar-1}
\]
for each individual $t \in R$. Since $R$ has cardinality 
$\gcd(r, q{-}1)$, 
combining with Equation~\eqref{eq:sumH} and Theorem~\ref{th:dimker},
we finally get the theorem.
\end{proof}

\begin{remark} 
\label{rem:boundzeroes}
For some choices of $\numvar$ and $\ue$, the bound of 
Theorem~\ref{th:boundzeroesOre} is sharp.
For example, it is the case for $\LL[X, Y; \text{id}, \Phi]$:
the bound is attained for instance with the polynomials 
$(X - a_1) \cdots (X - a_\degpol)$ where $a_1, \ldots, a_\degpol$ 
are pairwise distinct elements of $\F^\times$.
However, there are other parameters $(\numvar, \ue)$ for which the bound 
is not tight. In particular, when $\numvar = 2$ and $\ue = (r_1, r_2)$ with 
$r_1 < r_2$, $r_1 r_2 = r$ and $\gcd(r_1, r_2) = 1$ then, using the 
same techniques, one can show that
\[
  \sum_{\gamma \in H} \dim_{\F} \ker \epsilon_{\tilde \gamma}(f)
\leq r_2 \degpol \cdot (q{-}1)^{\numvar-1}
\]
improving then the upper bound of Theorem~\ref{th:boundzeroesOre} 
by a factor $r_1$.
It could be interesting to study these phenomena in more details.
\end{remark}

After this preparation, we are now ready to prove
Theorem~\ref{th:codeparameters}.

\begin{proof}[Proof of Theorem~\ref{th:codeparameters}]
Let $f \in \LL[\uX^{\pm 1}; \utheta]_{\leq \degpol}$.
It follows from Theorem~\ref{th:boundzeroesOre} that
\begin{align*}
  w_\srk\big(\epsilon(f)\big)
& = \sum_{\gamma \in H} \rk\: \epsilon_{\tilde \gamma}(f) \\
& = r \cdot (q{-}1)^\numvar 
  - \sum_{\gamma \in H} \dim_{\F} \ker \epsilon_{\tilde \gamma}(f) 
  \geq r \cdot (q{-}1)^{\numvar-1} \cdot (q-1-\degpol),
\end{align*}
hence the bound on the minimum distance. The same computation shows
in addition that $\epsilon$ is injective when restricted to the
subspace $\LL[\uX^{\pm 1}; \utheta]_{\leq \degpol}$. Therefore, the 
dimension of the code $\LRM(\ue;\,\degpol)$ is the same as the dimension
of $\LL[\uX^{\pm 1}; \utheta]_{\leq \degpol}$, \ie it is the number
of monomials in $\numvar$ variables of degree at most $\degpol$. A standard
computation indicates that it is the binomial coefficient
$\binom{\degpol+\numvar}\degpol$ as claimed.
\end{proof}

\subsection{Improving on the parameters}\label{ssec:improvedparams}

In what precedes, we have built our codes by restricting to Ore
polynomials of bounded total degree. This is certainly the most
natural thing to do; however, as we shall see, allowing
for more flexibility could sometimes lead to codes with better
parameters.

\begin{definition}
\label{def:support}
Let $f = \sum_{\uu \in \Z^\numvar} a_{\uu} \uX^{\uu} \in \LL[\uX^{\pm 1};
\utheta\big]$. The \emph{support} of $f$, denoted by $\Supp(f)$, is
the subset of $\Z^\numvar$ consisting of tuples $\uu$ for which $a_{\uu}$
does not vanish.

For a convex subset $C \subset \R^\numvar$, we let 
$\LL[\uX^{\pm 1}; \utheta]_C$ denote the subspace of 
$\LL[\uX^{\pm 1}; \utheta]$ consisting of Ore polynomials $f$
with $\Supp(f) \subset C$.
\end{definition}

Clearly $\LL[\uX^{\pm 1}; \utheta]_C$ is a $\LL$-vector subspace of 
$\LL[\uX^{\pm 1}; \utheta]$. A basis of it is given by the monomials
$\uX^{\uu}$ for $\uu$ running over the intersection $C \cap \Z^\numvar$.
In particular, it is finite dimensional when the latter intersection
is finite; this occurs for instance as soon as $C$ is compact.

\begin{definition}
\label{def:codesLRMrefined}
Let $C$ be a compact convex subset of $\R^\numvar$.
The \emph{linearized Reed-Muller code} associated to $\ue$ and 
$C$ is
$\LRM(\ue;\,C) = \epsilon\big(\LL[\uX^{\pm 1}; \utheta]_C\big)$.
\end{definition}

Beyond noticing that all the codes $\LRM(\ue;\,C)$ have length
$r{\cdot}(q{-}1)^\numvar$ (since they all sit in $\HH$), studying them
in full generality looks difficult.
There is however a special case for which a lot can be said. 
Let $\underline \uw = (\uw_1, \ldots, \uw_\numvar)$ be a basis of $L$ and 
let $S_{\underline \uw}$ be the simplex associated to it:
\[
  S_{\underline \uw}
  = \big\{\, \lambda_1 \uw_1 + \cdots + \lambda_\numvar \uw_\numvar \,\, : \,\,
      \lambda_i \in \R^+, \,\lambda_1 + \cdots + \lambda_\numvar \leq 1 \, \big\}.
\]
More generally, given an extra positive integer $\degpol$,
we consider its $\degpol$-dilation:
\[
  \degpol S_{\underline \uw}
  = \big\{\, \lambda_1 \uw_1 + \cdots + \lambda_\numvar \uw_\numvar \,\, : \,\,
      \lambda_i \in \R^+, \,\lambda_1 + \cdots + \lambda_\numvar \leq \degpol \, \big\}.
\]
When $\underline \uw$ and $\degpol$ vary, we obtain a family of codes
$\LRM(\ue;\,\degpol S_{\underline \uw})$ exhibiting quite nice properties.
To start with, we mention that a famous theorem of Ehrhart~\cite{Eh62} 
tells us that the number of integer points inside $\degpol S_{\underline \uw}$ 
varies quite regularly with respect to $\degpol$. More precisely, there 
exists a polynomial $P_{\underline \uw}(X)$, depending only on
$\underline \uw$ such that
$\Card\big(\degpol S_{\underline \uw} \cap \Z^\numvar\big) 
 = P_{\underline \uw}(\degpol)$ for all nonnegative integer $\degpol$. 
Besides, we know that $P_{\underline \uw}(X)$ has degree $\numvar$, that
its constant coefficient is $1$ and
that its leading coefficient is $\text{Vol}(S_{\underline \uw}) = 
\frac r{\numvar!}$. For $\degpol$ going to infinity, we then have the estimation
\begin{equation}
\label{eq:asympdim}
   \Card\big(\degpol S_{\underline \uw} \cap \Z^\numvar\big)
 = P_{\underline \uw}(\degpol) 
 = \frac{r {\cdot} \degpol^\numvar}{\numvar!} + O(\degpol^{\numvar-1}).
\end{equation}
A general upper bound on Ehrhart's polynomials is also known.
Precisely \cite[Theorem~7.a]{BeMc85} tells us that
\begin{equation}
\label{eq:bounddim}
  P_{\underline \uw}(\degpol) \leq
  \binom{\numvar+\degpol-1}{\numvar} \cdot r + \binom{\numvar+\degpol-1}{\numvar-1} = 
  \frac{(\degpol+1) \cdots (\degpol+\numvar-1) \cdot (r\degpol+\numvar)}{\numvar!}.
\end{equation}
for all nonnegative integer $\degpol$.

\begin{theorem}
\label{th:codeparameters2}
We keep the previous notation and assume that $\degpol$ is an integer
between $0$ and $q{-}2$. Then the dimension $k$ and the minimum
distance $\mindist$ of $\LRM(\ue;\,\degpol S_{\underline \uw})$ satisfy
\[
  k = \Card\big(\degpol S_{\underline \uw} \cap \Z^\numvar\big) 
    = P_{\underline \uw}(\degpol)
  \quad \text{and} \quad
  \mindist \geq r \cdot (q{-}1)^{\numvar-1} \cdot (q-1-\degpol).
\]
\end{theorem}

\begin{proof}
The proof follows the same pattern than that of
Theorem~\ref{th:codeparameters}, with significant simplifications.
For $i \in \{1, \ldots, \numvar\}$, write $Y_i = \uX^{\uw_i}$.
Let $f \in \LL[\uX^{\pm 1}; \utheta]_C$. Repeating the proof of
Lemma~\ref{lem:nrddegree}, we find that the reduced norm $\Nrd(f)$
has support included in $r\degpol S_{\underline\uw}$. When viewed as a
polynomial in $Y_1, \ldots, Y_\numvar$, it thus has total degree at
most $r\degpol$. Therefore, we can apply Proposition~\ref{prop:boundzeroes}
directly and get
\[
  \sum_{\gamma \in H} \ord_\gamma\big(\Nrd(f)\big) \leq
  \degpol \cdot (q{-}1)^{\numvar-1}.
\]
Using Theorem~\ref{th:dimker}, we obtain
\[
  \sum_{\gamma \in H} \dim_{\F} \ker \epsilon_{\tilde\gamma}(f) \leq
  \degpol \cdot (q{-}1)^{\numvar-1}
\]
and repeating the final argument of the proof of 
Theorem~\ref{th:codeparameters}, we conclude that the sum-rank
weight of $\epsilon(f)$ is at least $r{\cdot}(q{-}1)^{\numvar-1}{\cdot}(q-1-\degpol)$.
This gives the desired bound on the minimum distance. The formula for
the dimension follows as well.
\end{proof}

It follows from all what precedes that Equation~\eqref{eq:asympdim}
gives the asymptotic behaviour of the dimension of our codes
$\LRM(\ue;\,\degpol S_{\underline \uw})$. Comparing with the dimension
of $\LRM(\ue;\,\degpol)$, we see that we gain a factor $r$; indeed for
a fixed $\numvar$ and $\degpol$ going to infinity, we have $\binom{\numvar+\degpol}{\degpol}
\sim \frac{\degpol^\numvar}{\numvar!}$. Nonetheless, the lower bound on the minimum
distance remains the same. From this point of view, the codes
$\LRM(\ue;\,\degpol S_{\underline \uw})$ look much better than their
counterparts $\LRM(\ue;\,\degpol)$.

However, using $\LRM(\ue;\,\degpol S_{\underline \uw})$ might also have
some small disadvantages. One of them is that enumerating the points in
$\degpol S_{\underline\uw} \cap \Z^\numvar$ is not a straightforward task
(although efficient algorithms exist for this). Related to this,
Equation~\eqref{eq:asympdim} only provides asymptotic information
but is not applicable for small values of $\degpol$. In practice, working
with large values of $\degpol$ implies working over large finite fields 
as well (since $\degpol$ must be at most $q{-}2$), which could be an issue
in some situations.
Furthermore, for some applications where we are not only interesting
in optimizing the minimum distance, the codes $\LRM(\ue;\,\degpol)$ could
remain interesting as they seem to offer more diversity, in the
sense that the domain where we are picking the defining monomials
is not directly related to the lattice~$L$. Besides, after
Remark~\ref{rem:boundzeroes}, improving the estimation on the
minimum distance looks plausible in certain cases.

\subsection{Some examples}

\subsubsection{The ``almost commutative'' case}
\label{sssec:ac}

We focus on the case $\ue = (0,
\ldots, 0, 1)$ which turns out to be particularly interesting. In
this situation, we have an isomorphism
\[
  \LL[\uX^{\pm 1}; \utheta] \simeq
  \LL[X_1^{\pm 1}, \ldots, X_{\numvar-1}^{\pm 1}][X_\numvar^{\pm 1}; \Phi],
\]
so that the multivariate Ore polynomial algebra we work with is a 
univariate Ore Laurent polynomial ring over a classical Laurent polynomial
ring. Roughly speaking, the noncommutativity is entirely concentrated
on the last variable $X_\numvar$.
The lattice $L$ is easy to describe: if $(\ub_1, \ldots, \ub_\numvar)$
denotes the canonical basis of $\Z^\numvar$, $L$ is generated by the
vectors $\ub_1, \ldots, \ub_{\numvar-1}, r \ub_\numvar$. We consider the
associated family of codes 
\[
\LRM\big((0, \ldots, 0, 1);\,\degpol S_{(\ub_1, \ldots, \ub_{\numvar-1}, r\ub_\numvar)}\big),
\] 
for $\degpol$ varying in $\{1, \ldots, q{-}2\}$. 
Theorem~\ref{th:codeparameters2} provides estimations on the
parameters of these codes:
their dimension is given by the Ehrhart's polynomial
$P_{\underline w}(\degpol)$ and
their minimum distance is at least $r (q{-}1)^{\numvar-1} (q{-}1{-}\degpol)$.
In our particular case, we can be even more concrete and give a 
simple expression for the aforementioned Ehrhart's polynomial.

\begin{proposition}
\label{prop:ac:dim}
The dimension of the code 
$\LRM\big((0, \ldots, 0, 1);\,\degpol S_{(\ub_1, \ldots, \ub_{\numvar-1}, r\ub_\numvar)}\big)$
is
\[
  \binom{\numvar+\degpol}{\numvar} + (r-1) \cdot \binom{\numvar+\degpol-1}{\numvar} = 
  \frac{(\degpol+1) \cdots (\degpol+\numvar-1) \cdot (r\degpol+\numvar)}{\numvar!}
\]
\end{proposition}

\begin{proof}
We need to find the number of integer points inside the simplex
$\degpol S_{(\ub_1, \ldots, \ub_{\numvar-1}, r\ub_\numvar)}$, \ie to count the
integral nonnegative solutions $(x_1, \ldots, x_\numvar)$ of 
\begin{equation}
\label{eq:ac:Sw}
x_1 + \cdots + x_{\numvar-1} + \frac {x_\numvar} r \leq \degpol.
\end{equation}
We count separately the solutions $(x_1, \ldots, x_\numvar)$ with 
$x_\numvar \equiv u \pmod r$ for $u$ varying in $\{0, \ldots, r-1\}$.
Writing $x_\numvar = y_\numvar r + u$, Equation~\eqref{eq:ac:Sw} reduces to
$x_1 + \cdots + x_{\numvar-1} + y_\numvar \leq \degpol$ for $u = 0$ and
$x_1 + \cdots + x_{\numvar-1} + y_\numvar \leq \degpol{-}1$ for $u > 0$. 
The formula in the statement of the proposition follows.
\end{proof}

We observe that the dimension provided by Proposition~\ref{prop:ac:dim}
meets the upper bound~\eqref{eq:bounddim}; hence ``almost commutative''
linearized Reed--Muller codes are optimal regarding dimension.

The codes
$\LRM((0, \ldots, 0, 1);\,\degpol S_{(\ub_1, \ldots, \ub_{\numvar-1}, r\ub_\numvar)})$
exhibit actually another quite interesting feature. Indeed, given 
that the
variables $X_1, \ldots, X_{\numvar-1}$ are ``commutative'', we are not
obliged to inverse them and can evaluate them at $0$. Doing this,
we obtain an extended code
\[
\widetilde{\LRM}
\big((0, \ldots, 0, 1);\,\degpol S_{(\ub_1, \ldots, \ub_{\numvar-1}, r\ub_\numvar)}\big)
\]
with the following parameters:
\begin{itemize}
\item its length is $rq^{\numvar-1}{\cdot}(q{-}1)$,
\item its dimension is 
$(\degpol+1) \cdots (\degpol+\numvar-1){\cdot}(r\degpol+\numvar)/\numvar!$, and
\item its minimum distance is at least $rq^{\numvar-1}{\cdot}(q{-}1{-}\degpol)$.
\end{itemize}

\subsubsection{A concrete example}

We take $\numvar = 2$, $r = 4$ and $\ue = (3,2)$. By definition, the
lattice $L$ is the set of pairs $(x,y) \in \Z^2$ such that 
$3x + 2y \equiv 0 \pmod 4$, \ie $x \equiv 2y \pmod 4$. It is
represented on Figure~\ref{fig:lattice}.
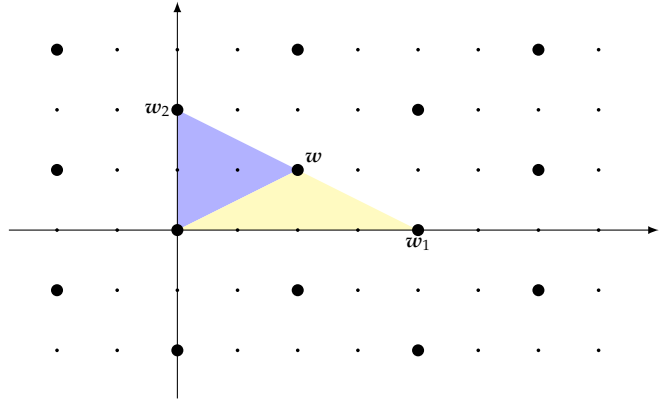
\begin{figure}
\centering
\begin{tikzpicture}[scale=0.8]
\fill[blue!30] (0,0)--(0,2)--(2,1);
\fill[yellow!30] (0,0)--(4,0)--(2,1);
\foreach \x in {-2,...,7} {
  \foreach \y in {-2,...,3} {
    \fill[black] (\x, \y) circle (0.3mm);
}}
\foreach \y in {-2,0,2} {
  \fill[black] (0, \y) circle (1mm);
  \fill[black] (4, \y) circle (1mm);
  \fill[black] (-2, \y+1) circle (1mm);
  \fill[black] (2, \y+1) circle (1mm);
  \fill[black] (6, \y+1) circle (1mm);
}
\node[below,scale=0.7] at (4,0) { $\uw_1$ };
\node[left,scale=0.7] at (0,2) { $\uw_2$ };
\node[above right,scale=0.7] at (2,1) { $\uw$ };
\draw[-latex] (-2.8,0)--(8,0);
\draw[-latex] (0,-2.8)--(0,3.8);
\end{tikzpicture}
\caption{The lattice $L$ for $\ue = (3, 2)$ and $r = 4$}
\label{fig:lattice}
\end{figure}
It contains the points $\uw_1 = (4,0)$, $\uw_2 = (0,2)$ and
$\uw = (2,1)$ and the families $(\uw_1, \uw)$ and $(\uw_2,
\uw)$ are two bases of $L$.

One can then form the corresponding codes
$\mathcal C_1(\degpol) := \LRM(\ue;\,\degpol S_{(\uw_1, \uw)})$ and
$\mathcal C_2(\degpol) := \LRM(\ue;\,\degpol S_{(\uw_2, \uw)})$ for
$\degpol < q{-}1$.
We infer from Theorem~\ref{th:codeparameters2} that they have minimum distance 
at least $4 \cdot (q{-}1) \cdot (q{-}1{-}\degpol)$. Moreover it is easy, in
our case, to compute exactly the polynomials $P_{(\uw_1, \uw)}$ and
$P_{(\uw_2, \uw)}$, which will eventually give the dimension of the
codes $\mathcal C_1$ and $\mathcal C_2$ respectively.
Indeed, we know that they take the form
$P_{(\uw_i, \uw)}(x) = 2 x^2 + a_i x + 1$
where $a_i$ is the unique unknown coefficient. One can find it by
evaluating at $x = 1$: counting the integer points insides
$S_{(\uw_i, \uw)}$, we find that $P_{(\uw_1, \uw)}(x) = 6$ and
$P_{(\uw_2, \uw)}(x) = 5$ from what we finally derive:
\[
P_{(\uw_1, \uw)}(x) = 2 x^2 + 3 x + 1
\quad \text{and} \quad
P_{(\uw_2, \uw)}(x) = 2 x^2 + 2 x + 1.
\]
The dimension of $\mathcal C_1(\degpol)$ (resp. of $\mathcal C_2(\degpol)$) is
then exactly $2 \degpol^2 + 3\degpol + 1$ (resp. $2 \degpol^2 + 2\degpol + 1$) for any~$\degpol$.
We observe that the former is greater than the latter, and that
both of them are larger than $\binom{\degpol+2} 2 = \frac{\degpol(\degpol+1)}2$ by
a factor of at least $r = 4$.

\begin{remark}
The Ehrhart's polynomial $P_{(\uw_1, \uw)}(x)$
factors as $(x+1)(2x+1)$ and meets the upper 
bound~\eqref{eq:bounddim}. This is in fact not a surprise because
the code $\mathcal C_1(\degpol)$ is isomorphic to a ``almost commutative''
linearized Reed--Muller code through the transformation
$X_1 \mapsto X_1$, $X_2 \mapsto X_1^{-2} X_2$.
\end{remark}

Although $(\uw_1, \uw_2)$ is not a basis of $L$, it makes sense
to consider the code
$\mathcal C(\degpol) := \LRM(\ue;\,\degpol S_{(\uw_1, \uw_2)})$. Using a
variation on 
Remark~\ref{rem:boundzeroes}, one can prove that its minimum 
distance is at least $4(q-1)(q-1 - 2\degpol)$. For $\degpol < \frac{q-1}2$, 
Ehrhart's result then implies that its dimension is $4\degpol^2 + 4\degpol + 1 
= (2\degpol+1)^2$.

\section{Embeddings into LAG codes}\label{sec:LAGcodes}

A nice feature of classical Reed--Muller codes is that they can
be embedded in large Reed--Solomon codes~\cite{PW04}, a property which
notably allows for efficient decoding.

In this section, we highlight a similar feature for the codes 
$\LRM(\ue;\,\degpol S_{\underline \uw})$ introduced in 
Subsection~\ref{ssec:improvedparams}. The main difference is that
the latter codes will not embed into linearized Reed--Solomon codes
but in some linearized Algebraic Geometry (LAG) codes, which were 
recently introduced by the same authors in~\cite{BC24}.
Unfortunately, no efficient decoding algorithm for LAG codes have
been designed so far. However, it looks feasible to extend standard
methods for decoding AG codes to the linearized setting; we hope to
come back on this question soon.

\subsection{Quick review on LAG codes}
\label{ssec:LAG}

We briefly review the theory of LAG codes as developed in~\cite{BC24}.
Since we will be using them in a very special case, we only focus
on this particular setting (which actually avoids talking 
about algebraic curves). On the contrary, for our purpose, we will
need to use LAG codes defined over extensions of $\F$.
That is why, for better consistency, we prefer as of now
considering a positive integer $\extlag$ and working over~$\Fn$.

We pick in addition a second positive integer $r$ and form the 
finite field $\LLn$, which is an extension of $\Fn$ of degree $r$.
Let $\Phi_\extlag : \LLn \to \LLn$, $x \mapsto x^{q^\extlag}$ be the relative
Frobenius of $\LLn/\Fn$, and let $\theta = \Phi_\extlag^e$ for some
integer $e$, coprime with $r$.
We also consider a new variable $Y$ and denote by $\LLn(Y)$ the
field of rational functions in $Y$. The morphism $\theta$ extends
naturally to an automorphism $\LLn(Y) \to \LLn(Y)$ by acting on 
the coefficients and letting $Y$ unchanged; in a slight abuse of
notation, we continue to call $\theta$ this extended morphism.

In order to define our LAG codes, we need extra data. First of
all, we consider a polynomial $P(Y) \in \Fn[Y]$ and, following
\cite{BC24}, we form the quotient
\begin{equation}
\label{eq:DP}
D_P = \LLn(Y)[T; \theta] / (T^r - P(Y)).
\end{equation}

\begin{lemma}
\label{lem:divisionalgebra}
We assume that the gcd of the orders of vanishing of $P(Y)$ 
at all points $y \in \Fn$ is coprime with $r$. Then $D_P$ is
a division algebra.
\end{lemma}

\begin{proof}
We recall from~\cite[\S 31]{Reiner75} that, to any central simple 
algebra $C$ over $\F(Y)$, one can associate a family of local invariants
$\inv_{\mathfrak p}(C) \in \Q/\Z$ indexed by the places $\mathfrak 
p$ of $F(Y)$. They satisfy in addition the following two properties:
\begin{enumerate}[(i)]
\item the invariants of $C$ are the same than the invariants of
$M_\extlag(C)$ for all~$\extlag > 0$ (see \cite[\S 28]{Reiner75})
\item if $C$ has dimension $s^2$ over $F(Y)$, the invariants of
$C$ are all in $s^{-1} \Z/\Z$ (see \cite[Theorem 29.22]{Reiner75}).
\end{enumerate}
Besides, the invariants of $D_P$ are easy to write down; indeed,
it follows from~\cite[Equation~(31.7)]{Reiner75} that
\[
  \inv_{\mathfrak p}(D_P) = \frac{v_{\mathfrak p}\big(P(Y)\big)} r
  \in \Q / \Z
\]
where $v_{\mathfrak p}$ is the normalized valuation associated to
the place $\mathfrak p$. In particular, when $\mathfrak p = 
\mathfrak p_y$ is the place corresponding to a point $y \in \Fn$, 
the invariant $\inv_{\mathfrak p_y}(D_P)$ is $-\ord_y(P(Y))/r 
\text{ mod } \Z$.

We now invoke the Artin--Wedderburn theorem~\cite[Theorem 7.4]{Reiner75} 
which tells us 
that $D_P$ has to be isomorphic to a matrix algebra over a division 
algebra $\Delta$ over $\F(Y)$. Write $\dim_{\F(Y)} \Delta = s^2$. 
Using the properties recalled earlier, we find that 
\[
  -\frac{\ord_y(P(Y))} r \equiv
  \inv_{\mathfrak p_y}(D_P) = 
  \inv_{\mathfrak p_y}(\Delta) \in s^{-1} \Z / \Z
\]
for all $y \in \Fn$. In other words $s{\cdot}\ord_y(P(Y)) \in
r \Z$ for all $y \in \Fn$. Since the gcd of $\ord_y(P(Y))$ (for
$y$ running over $\Fn$) is
coprime with $r$, Bézout's theorem shows that $s$ must lie in 
$r \Z$ as well. Thus $s = r$, which further implies by comparing
dimensions that $D_P = \Delta$ and finally that $D_P$ is itself a 
division algebra.
\end{proof}

From now on, we assume that the hypothesis of
Lemma~\ref{lem:divisionalgebra} is fulfilled.
Another important ingredient we need is the notion of Riemann--Roch 
space inside $D_P$. For our purpose, we will only need
them in a particular case, so we restrict ourselves to this one (see \cite[Subsection 2.2]{BC24} for the general definition).

\begin{definition}
\label{def:RR}
To each nonnegative integer $v$, we attach
the \emph{Riemann--Roch space} $\Lambda_P(v)$ defined as the 
$\LLn$-vector subspace of $D_P$ consisting of Ore polynomials
of the form 
$\sum_{i=0}^{r-1} \frac{u_i(Y)}{v_i(Y)} T^i$
where, for all $i$, the polynomials $u_i(Y), v_i(Y) \in \LLn(Y)$
are subject to the following conditions:
\begin{itemize}
\item $\gcd(u_i(Y), v_i(Y)) = 1$,
\item $r \cdot\big(\!\deg u_i(Y) - \deg v_i(Y)\big) + i \cdot \deg P(Y) \leq v$, 
\item $v_i(Y)^r$ divides $P(Y)^i$.
\end{itemize}
\end{definition}

\begin{remark}
We note that $\Lambda_P(v)$ contains the following simpler
space
\[
  \tilde \Lambda_P(v) =
  \left\{\, 
    \sum_{i=0}^{r-1} u_i(Y) T^i \, : \, u_i(Y) \in \LLn[Y], \,
    r{\cdot}\!\deg u_i(Y) + i{\cdot}\!\deg P(Y) \leq v
  \,\right\},
\]
which is in fact the one we will work with afterwards.
\end{remark}

We also define $\Lambda_P = \Lambda_P(\infty)$ as the union of 
the $\Lambda_P(v)$ when $v$ varies; it is a subalgebra of $D_P$.
We consider elements $y_1, \ldots, y_s \in \Fn$ such that
$P(y_i) \neq 0$ for all $i$. By \cite[Lemmas~1.1~\&~3.2]{BC24}, 
the latter assumption implies the existence of isomorphisms
\begin{equation}
\label{eq:evalisomY}
  \eta_i : \Lambda_P / (Y{-}y_i) \Lambda_P
  \stackrel\sim\longrightarrow
  \End_{\Fn}(\LLn).
\end{equation}
We combine them into a unique multievaluation map
\[
  \eta = (\eta_1, \ldots, \eta_s) : 
  \Lambda_P \longrightarrow \End_{\Fn}(\LLn)^s.
\]

\begin{definition}
The \emph{linearized Algebraic Geometry} code attached to the
previous data is
\[
  \LAG\big(P(Y);\, v;\, y_1, \ldots, y_s\big) =
  \eta\big(\Lambda_P(v)\big).
\]
\end{definition}

By definition, the code $\LAG\big(P(Y);\, v;\, y_1, \ldots, y_s\big)$
sits in $\End_{\Fn}(\LLn)^s$. The latter is a vector space over 
$\LLn$ of dimension $sr$, the length of the code.
We notice moreover that the ambient space $\End_{\Fn}(\LLn)^s$
is naturally equipped with the sum-rank metric; hence it makes
sense to talk about the minimum sum-rank distance of the code
$\LAG\big(P(Y);\, v;\, y_1, \ldots, y_s\big)$.
It follows from \cite[Theorem~3.5]{BC24} that this minimum 
distance is at least
\begin{equation}
\label{eq:LAG:designed}
  \mindist^\star\big(\LAG\big(P(Y);\, v;\, y_1, \ldots, y_s\big)\big) := sr - v.
\end{equation}
In what follows, this lower bound will be called the 
\emph{designed minimum distance} of $\LAG\big(P(Y);\, v;\, y_1, 
\ldots, y_s\big)$.

\subsection{Relating the centre to a univariate rational function field}

We now come back to the setting of Section~\ref{sec:thecode}: we
consider the multivariate Ore algebra $\LL[\uX^{\pm 1}; \utheta]$ where
$\uX = (X_1, \ldots, X_\numvar)$ and $\utheta = (\theta_1, \ldots,
\theta_\numvar)$ with $\theta_i = \Phi^{e_i}$ and $\Phi : \LL \to 
\LL$ is the Frobenius $x \mapsto x^q$.
We assume as usual that $\gcd(e_1, \ldots, e_\numvar, r) = 1$ and
write $\ue = (e_1, \ldots, e_\numvar)$. We recall that the centre
of $\LL[\uX^{\pm 1}; \utheta]$ is $\F[\uX^L]$ where $L$ is the
lattice
\[
  L=\big\{\,\uu=(u_1,\dots,u_\numvar)\in\Z^\numvar \,|\, \ue{\cdot}\uu\in r\Z\,\big\}.
\]
From now, we fix a $\Z$-basis $\underline \uw = (\uw_1, \ldots, \uw_\numvar)$ of $L$.
This choice gives rise to an isomorphism between $\F[\uX^L]$ 
and the multivariate Laurent polynomial ring $\F[Z_1^{\pm 1}, 
\ldots, Z_\numvar^{\pm 1}]$ where the variable $Z_i$ corresponds to
$\uX^{\uw_i}$. To shorten notation, we use again bold symbols
for tuples and set $\uZ = (Z_1, \ldots, Z_\numvar)$ 
and $\F[\uZ^{\pm 1}] = \F[Z_1^{\pm 1}, \ldots, Z_\numvar^{\pm 1}]$.

The first step in our construction is to relate $\F[\uX^L] 
\simeq \F[\uZ^{\pm 1}]$ to the univariate rational function 
field $\Fn(Y)$ for any $\extlag \geq \numvar$.
In order to do so, we choose a basis $(b_1, \ldots, b_\extlag)$ of
$\Fn$ over $\F$. For each $i \in \{1, \ldots, \extlag\}$, we define
the $\F$-linear form
\[
  \beta_i : \Fn \to \F, \quad y \mapsto \text{Tr}_{\Fn/\F}(b_i y)
\]
where $\text{Tr}_{\Fn/\F}$ is the trace map. It is a 
well-known fact (see \emph{e.g.} \cite[Theorem VI.5.2]{Lang3}) 
that the $\beta_i$'s 
form a basis of $\Hom_{\F}(\Fn,\F)$. Hence the map
$\beta = (\beta_1, \ldots, \beta_\extlag) : \Fn \to \F^\extlag$
is a $\F$-linear isomorphism. Let $E$ be the inverse
image by $\beta$ of $\F^\numvar \times \{0\}^{\extlag-\numvar} \subset \F^\extlag$,
\ie $E = \bigcap_{j > \numvar} \ker \beta_j$. By what precedes, it is 
a $\F$-vector space of dimension $\numvar$; more precisely,
$\beta_{\leq \numvar} := (\beta_1, \ldots, \beta_\numvar)$ induces an 
isomorphism between $E$ and $\F^\numvar$.
The following lemma asserts that $\beta_{\leq \numvar}$ is a
polynomial function of controlled degree.

\begin{lemma}
\label{lem:Bi}
There exist polynomials $B_1(Y), \ldots, B_\numvar(Y) \in \Fn[Y]$ of
degree at most $q^{\numvar-1}$ such that $\beta_i(y) = B_i(y)$
for all $y \in E$ and all $i \in \{1, \ldots, \numvar\}$.
\end{lemma}

\begin{proof}
We consider the ring of univariate Ore polynomials $\Fn[U; \Phi]$
where $\Phi : x \mapsto x^q$ is the Frobenius and $U$ is again a
new variable. We recall from~\cite{Ore33} that it is left and right
Euclidean. In particular, it is left and right principal and it
admits left and right gcd and lcm.
We consider the standard evaluation morphism
\[
  \epsilon :
  \Fn[U; \Phi] \longrightarrow
  \End_{\F}(\Fn), \quad f \mapsto f(\Phi).
\]
Each $\beta_i$ defines an element in $\End_{\F}(\Fn)$ and we
check that $\beta_i = \epsilon(T_i)$ with
\begin{align*}
  T_i 
& = \big(1 + U + U^2 + \cdots + U^{\extlag-1}\big) \cdot b_i \\
& = b_i + b_i^q U + b_i^{q^2} U^2 + \cdots + b_i^{q^{\extlag-1}} U^{\extlag-1}.
\end{align*}
Let $\mathcal I$ be the left ideal of $\End_{\F}(\Fn)$ consisting 
of endomorphisms vanishing on $E$. Coming back to the definition of 
$E$, we infer that $\epsilon^{-1}(\mathcal I)$ is the principal
left ideal generated by
$T_{>\numvar} = \rgcd(T_{\numvar+1}, \ldots, T_\extlag)$,
where the notation $\rgcd$ refers to the right gcd. Hence, 
$\epsilon$ induces a $\Fn$-linear isomorphism
\[
  \bar\epsilon :
  \Fn[U; \Phi] \,/\, \Fn[U; \Phi]{\cdot} T_{>\numvar}
  \stackrel\sim\longrightarrow
  \Hom_{\F}(E, \Fn).
\]
By comparing dimensions, we derive that the degree of $T_{>\numvar}$ 
is equal to~$\numvar$. 
Let $i \in \{1, \ldots, \numvar\}$.
The restriction of $\beta_i$ to $E$, namely ${\beta_i}_{|E}$, 
can be seen as an element of $\Hom_{\F}(E, \Fn)$. Let 
$\tilde B_i$ by the unique representative of 
$\bar\epsilon^{-1}\big({\beta_i}_{|E}\big)$ of degree at most 
$\numvar{-}1$. By definition, we have 
$\tilde B_i(\Phi)(y) = \beta_i(y)$
for all $y \in E$. 
Finally, given that $\Phi$ itself is a polynomial function of
degree $q$, we find that $\tilde B_i(\Phi)$ is a polynomial
function of degree $q^{\deg \tilde B_i} \leq q^{\numvar-1}$. This
concludes the proof by setting $B_i = \tilde B_i(\Phi)$.
\end{proof}

\begin{remark}
\label{rem:Bi}
The polynomial $B_i(Y)$ vanishes on the space 
$E_i := E \cap \ker \beta_i$, which has cardinality $q^{\numvar-1}$.
Therefore, there must exist a nonzero constant $c_i \in \Fn^\times$
such that
\[
B_i(Y) = c_i \cdot \prod_{y \in E_i} (Y-y).
\]
In particular, we notice that $B_i(Y)$ is separable and has exactly
degree $q^{\numvar-1}$.
\end{remark}

We use the polynomials $B_1(Y), \ldots, B_\numvar(Y)$ of 
Lemma~\ref{lem:Bi} to build the following morphism of $\F$-algebras
\[\begin{array}{rcl}
  \zeta : \quad 
    \F[\uX^L] \simeq \F[\uZ^{\pm 1}] 
  & \longrightarrow 
  & \Fn\left[Y, \frac 1{B(Y)}\right] 
    \subset \Fn(Y) \smallskip \\
  Z_i & \mapsto & B_i(Y)
\end{array}\]
where we have set $B(Y) = B_1(Y) \cdots B_\numvar(Y)$ for simplicity.
We now briefly describe the action of $\zeta$ on evaluation
points. 
To start with, we recall from Subsection~\ref{subsec:evalOre} that an evaluation point for $\F[\uX^L]$ is
given by a group homomorphism $\gamma : L \to \F^\times$. Transferring
it \emph{via} the identification $\F[\uX^L] \simeq \F[\uZ^{\pm 1}]$,
it simply corresponds to the multievaluation point
$\big(\gamma(\uw_1), \ldots, \gamma(\uw_\numvar)\big)$. From this, it
is routine to check that the following diagram commutes
\begin{equation}
\label{diag:diagevalpoints}
\raisebox{5ex}{
\xymatrix @C=8ex {
  \F[\uX^L] \ar[r]^-{\zeta} \ar[d]_-{\hspace{12ex}\epsilon_\gamma} 
    & \Fn\left[Y, \frac 1{B(Y)}\right]
      \ar[d]^-{Y \mapsto \beta_{\leq \numvar}^{-1}
      (\gamma(\uw_1), \ldots, \gamma(\uw_\numvar))} \\
  \F \ar[r] & \Fn
}}
\end{equation}
where the bottom arrow is the canonical inclusion. In other words,
we have shown that $\gamma$ corresponds to the evaluation point
$\beta_{\leq \numvar}^{-1} \big(\gamma(\uw_1), \ldots, \gamma(\uw_\numvar)\big)$.

\subsection{Extension to the Ore algebra}
\label{ssec:LAG:extOre}

So far, we have constructed a morphism $\zeta : \F[\uX^L]
\to \Fn(Y)$ which, in some sense, relates the multivariate case
to the univariate one. The next step in the construction is to 
extend $\zeta$ to the Ore algebra $\LL[\uX^{\pm 1}; \utheta]$.
For this, we recall that the map $\uu \mapsto \ue{\cdot}\uu$
induces a group isomorphism $\Z^\numvar/L \to \Z/r\Z$. We choose a
vector $\uv \in \Z^\numvar$ such that $\ue{\cdot}\uv \equiv 1 \pmod r$.
By what precedes $\Z^\numvar$ is generated as a group by $L$ and $\uv$.
Therefore $\LL[\uX^{\pm 1}; \utheta]$ is generated as an algebra
by its centre $\F[\uX^L]$ and the monomial $\uX^{\uv}$. Besides,
by our choice of $\uv$, the commutation relation
$\uX^{\uv} \cdot a = \Phi(a) \cdot \uX^{\uv}$ holds for all $a 
\in \LL$. 
It follows from this observation that we have a surjective
morphism
\[
\begin{array}{rcl}
  \F[\uX^L] \otimes_{\F} \LL[T;\Phi] 
    & \longrightarrow & \LL[\uX^{\pm 1}; \utheta] \smallskip \\
  T & \mapsto & \uX^{\uv}
\end{array}.
\]
The
kernel of this map obviously contains the ideal generated by
$T^r{-}\uX^{r \uv}$ (note that $r \uv \in L$); by comparing
dimensions, we conclude that the reverse inclusion also holds
true. Consequently, we get an isomorphism of $\F$-algebras
\[
  \alpha : \,
  \F[\uX^L] \otimes_{\F} \LL[T;\Phi]
    / (T^r - \uX^{r \uv})
  \stackrel\sim\longrightarrow
  \LL[\uX^{\pm 1}; \utheta].
\]
Composing the inverse of $\alpha$ by $\zeta \otimes \text{id}$,
we obtain a second morphism
\[ \textstyle
  \iota : \,
  \LL[\uX^{\pm 1}; \utheta]
  \longrightarrow
  \Fn\left[Y, \frac 1{B(Y)}\right] \otimes_{\F} \LL[T;\Phi] / (T^r - P(Y))
\]
where, by definition, $P(Y) = \zeta(\uX^{r \uv}) \in \Fn[Y]$.

From now on, we assume that, in addition to be greater or equal
to $\numvar$, the integer $\extlag$ is chosen in such a way that $\gcd(\extlag, r)
= 1$. The extensions $\Fn$ and $\LL$ are then linearly disjoint
over $\Fq$, implying that the tensor product $\Fn \otimes_{\F} \LL$
is a field. Since the latter has cardinality $q^{rn}$, it must be isomorphic
to $\LLn$. Therefore, the codomain of $\iota$ is isomorphic to
\begin{equation} \textstyle
\label{eq:LambdaP}
  \LLn\left[Y, \frac 1{B(Y)}\right][T;\theta] / (T^r - P(Y))
\end{equation}
where $\theta$ is the automorphism acting as the identity on the
subfield $\Fn$ and as $x \mapsto x^q$ on the subfield $\LL$. From
these properties, we infer that $\theta = \Phi_\extlag^{\extlag'}$ where $\extlag'$
is a multiplicative inverse of $\extlag$ modulo $r$ and we recall that
$\Phi_\extlag$ is the relative Frobenius of $\LLn/\Fn$, \ie $\Phi_\extlag : x
\mapsto x^{q^\extlag}$. In Equation~\eqref{eq:LambdaP}, we recognize an
integral version of the algebra
\[
  D_P = \LLn(Y)[T;\theta] / (T^r - P(Y))
\]
already considered in Subsection~\ref{ssec:LAG}; precisely, the 
algebra of Equation~\eqref{eq:LambdaP} is $\Lambda_P\left[\frac 1{B(Y)}\right]$
where $\Lambda_P$ was introduced right after Definition~\ref{def:RR}.

\begin{lemma}
$D_P$ is a division algebra.
\end{lemma}

\begin{proof}
Let $\mu_1, \ldots, \mu_\numvar \in \Z$ be the coordinates of $r \uv$ 
in the basis $(\uw_1, \ldots, \uw_\numvar)$, so that we have
$r \uv = \mu_1 \uw_1 + \cdots + \mu_\numvar \uw_\numvar$. Taking the scalar
product of this equality by $\ue$, we find the relation
\[
  \uv{\cdot}\ue 
  = \mu_1 \cdot \frac{\uw_1{\cdot}\ue} r 
    + \cdots +
    \mu_\numvar \cdot \frac{\uw_\numvar{\cdot}\ue} r.
\]
Note that each quotient $(\uw_i{\cdot}\ue)/r$ is an integer,
given that $\uw_i$ lies in $L$. Besides, by our choice of $\uv$,
we know that $\uv{\cdot}\ue \equiv 1 \pmod r$. Hence, we deduce
that $\gcd(\mu_1, \ldots, \mu_\numvar, r) = 1$.

Thanks to Lemma~\ref{lem:divisionalgebra}, it suffices to find 
elements $y_1, \ldots, y_\numvar \in \Fn$ such that the order of vanishing
of $P(Y)$ at $y_i$ is $\mu_i$ for all $i$.
For this, note that $P(Y) = \zeta(\uX^{r\uv}) = 
B_1(Y)^{\mu_1} \cdots B_\numvar(Y)^{\mu_\numvar}$.
Moreover, we know from
Remark~\ref{rem:Bi} that the $B_i(y)$'s are all separable. Therefore, 
it is enough to find $y_i \in \Fn$ which is a root of $B_i(Y)$, but 
not a root of the other $B_j(Y)$'s. An element satisfying these
requirements is, for example, $y_i = \beta_{\leq \numvar}^{-1}
(1, \ldots, 1, 0, 1, \ldots, 1)$ where the $0$ is in the $i$th 
position.
\end{proof}

We now aim at comparing evaluation points in the spirit of the
diagram~\eqref{diag:diagevalpoints}. In order to do so, we first
recall that, each time we are given an element $y \in \Fn$
such that $P(y) \neq 0$, we have an evaluation map
$\eta_y \:: \Lambda_P \longrightarrow \End_{\Fn}\big(\LLn\big)$
whose kernel is the principal twosided ideal generated by
$Y{-}y$.
We note that $\eta_y$ maps $B(Y)$ to the scalar multiplication
by $B(y)$. Thus, if $y$ is chosen outside the roots of $B(Y)$, 
the morphism $\eta_y$ extends to a second homomorphism of
$\Fn$-algebras
\[ \textstyle
  \Lambda_P\left[\frac 1{B(Y)}\right] 
  \longrightarrow \End_{\Fn}\big(\LLn\big)
\]
that, in a slight abuse of notation, we continue to denote by
$\eta_y$.

On the other hand, we recall from Subsection~\ref{subsec:evalOre} 
that, whenever we are given a group homomorphism $\gamma : L \to 
\F^\times$ together with a prolongation $\tilde \gamma : \Z^\numvar
\to \LL^\times$ satisfying the axiom~\eqref{cocycle}, we can
build an evaluation map
\[
  \epsilon_{\tilde \gamma} \:: \LL[\uX^{\pm 1}; \utheta]
  \longrightarrow \End_{\F}\big(\LL\big).
\]
By Theorem~\ref{th:evalisom}, we know moreover that the kernel of 
$\epsilon_{\tilde \gamma}$ is the twosided ideal generated by
the elements $\uX^{\uu} - \gamma(\uu)$, $\uu \in L$.
In particular, it only depends on $\gamma$, and not on the 
choice of the prolongation $\tilde \gamma$.

\begin{lemma}
\label{lem:diagram}
Keeping the previous notation and setting $y = 
\beta_{\leq \numvar}^{-1} (\gamma(\uw_1), \ldots, \gamma(\uw_\numvar))$, 
the diagram
\[
\xymatrix @C=8ex {
  \LL[\uX^{\pm 1}; \utheta] \ar[r]^-{\iota} \ar[d]_-{\epsilon_{\tilde \gamma}} 
    & \Lambda_P\left[\frac 1{B(Y)}\right] \ar[d]^-{\eta_y} \\
  \End_{\F}(\LL) \ar[r] & \End_{\Fn}(\LLn)
}
\]
commutes up to conjugacy, \ie there exists a $\Fn$-linear 
automorphism $h_{\tilde \gamma} : \LLn \to \LLn$ such that 
\[
  \eta_y \big(\iota(f)\big) = 
  h_{\tilde\gamma}^{-1} 
  \circ \big( \text{\rm id}_{\Fn} \otimes \epsilon_{\tilde \gamma}(f)\big)
  \circ h_{\tilde\gamma}
\]
for all $f \in \LL[\uX^{\pm 1}; \utheta]$.
\end{lemma}

\begin{proof}
It follows from Theorem~\ref{th:evalisom} that,
after scalar extension to $\Fn$, $\epsilon_{\tilde \gamma}$ induces
an isomorphism of $\Fn$-algebras
\begin{equation}
\label{eq:evalpts:isom1}
  \text{\rm id}_{\Fn} \otimes \epsilon_{\tilde \gamma}(f) : \:
  \LLn[\uX^{\pm 1};\utheta] / \mathfrak{m}_\gamma \LLn[\uX^{\pm 1};\utheta] 
  \stackrel\sim\longrightarrow \End_{\Fn}(\LLn)
\end{equation}
where we recall that $\m_\gamma = \ker \epsilon_\gamma$.
On the other hand, we deduce from Equation~\eqref{eq:evalisomY}
that $\eta_y$ induces an isomorphism of $\Fn$-algebras
\[ \textstyle
  \eta_y : \,
  \Lambda_P\left[\frac 1{B(Y)}\right] / (Y-y) \Lambda_P\left[\frac 1{B(Y)}\right]
  \stackrel\sim\longrightarrow \End_{\Fn}(\LLn).
\]
Looking at the diagram~\eqref{diag:diagevalpoints}, we find that
the inverse image by $\iota$ of the ideal generated by $Y{-}y$ is
the ideal generated by $\m_\gamma$. As a consequence, the composite
$\eta_y \circ \iota$ induces another isomorphism
\begin{equation}
\label{eq:evalpts:isom2}
  \eta_y \circ \iota: \:
  \LLn[\uX^{\pm 1};\utheta] / \mathfrak{m}_\gamma \LLn[\uX^{\pm 1};\utheta] 
  \stackrel\sim\longrightarrow \End_{\Fn}(\LLn).
\end{equation}
We conclude by invoking the Skolem--Noether theorem, which ensures
that the isomorphisms~\eqref{eq:evalpts:isom1}
and~\eqref{eq:evalpts:isom2} have to be conjugated by an element
in $\text{GL}_{\Fn}(\LLn)$.
\end{proof}

\subsection{Comparison of codes}

After this long preparation, we are now ready to relate the
code $\LRM(\ue;\,\degpol S_{\underline \uw})$ to some well-chosen LAG
code. To start with, let us recall briefly from Subsection
\ref{ssec:construction} that the former is defined as the image 
of $\LL[\uX^{\pm 1}; \utheta]_{\degpol S_{\underline \uw}}$ under the
multievaluation morphism
\[
\begin{array}{rcl}
  \epsilon :\,
  \LL[\uX^{\pm 1}; \utheta]
& \longrightarrow 
& \prod_{\gamma \in H} \End_{\F}(\LL) 
  \smallskip \\
  f & \mapsto & \epsilon_{\tilde \gamma}(f)
\end{array}
\]
where $H = \Homgrp(L, \F^\times)$ and for each $\gamma \in H$,
we have chosen a cocycle $\tilde \gamma : \Z^\numvar \to \LL^\times$
extending $\gamma$. Recall also that, in what precedes, 
$\LL[\uX^{\pm 1}; \utheta]_{\degpol S_{\underline \uw}}$ is the
$\LL$-linear subspace of $\LL[\uX^{\pm 1}; \utheta]$ spanned
by the monomial $\uX^{\uu}$, $\uu \in \degpol S_{\underline \uw}$.

We now contemplate the commutative diagram of 
Lemma~\ref{lem:diagram}. Taking the product over all 
$\gamma \in H$, the
following diagram also commutes up to conjugacy:
\[
\xymatrix @C=8ex {
  \LL[\uX^{\pm 1}; \utheta] \ar[r]^-{\iota} \ar[d]_-{\epsilon} 
    & \Lambda_P\left[\frac 1{B(Y)}\right] \ar[d]^-{\eta = (\eta_y)_{y \in I}} \\
  \displaystyle
  \prod_{\gamma \in H} \End_{\F}(\LL)
  \ar[r] & 
  \displaystyle
  \prod_{y \in I} \End_{\Fn}(\LLn)
}
\]
where the index set $I$ is $\beta_{\leq \numvar}^{-1}\big((\F^\times)^\numvar\big)$ 
and the arrow on the bottom is induced by the correspondence $H \simeq
I$, $\gamma \mapsto y 
= \beta_{\leq \numvar}^{-1}(\gamma(\uw_1), \ldots, \gamma(\uw_\numvar))$.
As a consequence, if we can prove that $\iota$ takes 
$\LL[\uX^{\pm 1}; \utheta]_{\degpol S_{\underline \uw}}$ to some explicit
Riemann--Roch space inside $\Lambda_P\left[\frac 1{B(Y)}\right]$, we 
will infer a relation between $\LRM(\ue;\,\degpol S_{\underline \uw})$
and a suitable LAG code. This is achieved in the next lemma.

\begin{lemma}
We have
$\iota\big(\LL[\uX^{\pm 1}; \utheta]_{\degpol S_{\underline \uw}}\big)
 \subset \Lambda_P\big(q^{\numvar-1}r\degpol\big)$.
\end{lemma}

\begin{proof}
By linearity, it is enough to prove that
$\iota(\uX^{\uu}) \in \Lambda_P\big(q^{\numvar-1}r\degpol\big)$ for all $\uu \in \degpol 
S_{\underline \uw}$. For this, we write
$\uu = \uw + \lambda \uv$
with $\uw \in L$, $\lambda \in \{0, \ldots, 
r{-}1\}$, and where $\uv$ is the special
vector we fixed at the beginning of Subsection~\ref{ssec:LAG:extOre}. 
We decompose $\uw$ and $r \uv$, which are both elements of $L$, on
the basis $(\uw_1, \ldots, \uw_\numvar)$, namely we write
$\uw = \lambda_1 \uw_1 + \cdots + \lambda_\numvar \uw_\numvar$ and
$r \uv = \mu_1 \uw_1 + \cdots + \mu_\numvar \uw_\numvar$, where the $\lambda_i$
and the $\mu_i$ are integers.
From these equalities, we derive
\[
  \uu = \left(\lambda_1 - \frac{\lambda}{r} \mu_1\right) \uw_1
      + \cdots +
        \left(\lambda_\numvar - \frac{\lambda}{r} \mu_\numvar\right) \uw_\numvar
\]
and the assumption that $\uu \in \degpol S_{\underline w}$ tells us that
\begin{equation}
\label{eq:boundval}
  r\lambda_i - \lambda \mu_i \geq 0 \quad \text{for all } i
  \qquad \text{and} \qquad
  \sum_{i=1}^\numvar r \lambda_i - \lambda \mu_i \leq r\degpol.
\end{equation}
On the other hand, it follows from the definition of $\iota$ that
\[
  \iota(\uX^{\uu}) =
  \iota(\uX^{\uw_1})^{\lambda_1} \cdots
  \iota(\uX^{\uw_\numvar})^{\lambda_\numvar} \cdot \iota(\uX^{\uv})^\lambda
  = B_1(Y)^{\lambda_1} \cdots B_\numvar(Y)^{\lambda_\numvar} \cdot T^\lambda.
\]
Hence, in order to prove that $\iota(\uX^{\uu}) \in 
\Lambda_P\big(q^{\numvar-1}r\degpol\big)$, we have to check that
\[
  - \lambda{\cdot}\deg P(Y) +  r{\cdot}\sum_{i=1}^\numvar \lambda_i \deg B_i(Y)
 \leq q^{\numvar-1} r\degpol.
\]
This follows directly from Equation~\eqref{eq:boundval}  after
remembering that $\deg B_i(Y) \leq q^{\numvar-1}$ for all $i$ (see
Lemma~\ref{lem:Bi}) and that 
$P(Y) = \iota(\uX^{r\uv}) = B_1(Y)^{\mu_1} \cdots B_\numvar(Y)^{\mu_\numvar}$,
which ensures that $\deg P(Y) = \sum_{i=1}^\numvar \mu_i \deg B_i(Y)$.
\end{proof}

Finally, we have proved the following theorem.

\begin{theorem}
\label{th:embedLAG}
With the previous notation, the code
$\Fn\otimes_{\F}\LRM\big(\ue; \degpol S_{\underline \uw}\big)$
is isomorphic to a subcode of
$\LAG\big(P(Y); q^{\numvar-1}r\degpol; I\big)$.
Moreover the isomorphism is explicit and it preserves the
sum-rank distance.
\end{theorem}

We emphasize that the theorem is valid for any integer $\extlag \geq \numvar$
which is coprime with~$r$; in particular, we can always choose $\extlag$
in the range $[\numvar, \numvar{+}r]$.
On a different note, it is also instructive to compare the 
designed minimum distances of $\LRM\big(\ue; \degpol S_{\underline 
\uw}\big)$ and $\LAG\big(P(Y); q^{\numvar-1}r\degpol; I\big)$. After
Theorem~\ref{th:codeparameters2} and Equation~\eqref{eq:LAG:designed},
we have the following explicit values for them:
\begin{align*}
  \mindist^\star\big(\LRM\big(\ue; \degpol S_{\underline \uw}\big)\big) & 
  = (q{-}1)^\numvar r - (q{-}1)^{\numvar-1}r\degpol, \\
  \mindist^\star\big(\LAG\big(P(Y);\, q^{\numvar-1}r\degpol;\, I\big)\big) & 
  = (q{-}1)^\numvar r - q^{\numvar-1}r\degpol.
\end{align*}
We observe that they almost coincide apart from the factor $(q{-}1)^{\numvar-1}$ 
which is replaced by $q^{\numvar-1}$ in the second case. The conclusion is 
that the embedding of $\LRM\big(\ue; \degpol S_{\underline \uw}\big)$
in $\LAG\big(P(Y); q^{\numvar-1}r\degpol; I\big)$ does not alter too much
the (designed) minimum distance. Hence, any efficient decoder
for linearized Algebraic Geometry codes will provide a barely 
less efficient decoder for linearized Reed--Muller codes of the 
type $\LRM\big(\ue; \degpol S_{\underline \uw}\big)$.

\begin{remark}
In the ``almost commutative'' case (see \S \ref{sssec:ac}), the
embedding of Theorem~\ref{th:embedLAG} extends to an embedding
(up to conjugacy)
\[
\Fn \otimes_{\F} \widetilde{\LRM}\big(\ue; \degpol S_{\underline \uw}\big)
\hookrightarrow
\LAG\big(P(Y); q^{\numvar-1}r\degpol; \tilde I\:\big)
\]
with $\tilde I = \beta_{\leq \numvar}^{-1}(\F^{\numvar-1} \times \F^\times)$. The
designed minimum distance of the involved LAG code is now
\[
  \mindist^\star\big(\LAG\big(P(Y);\, q^{\numvar-1}r\degpol;\, \tilde I\:\big)\big) 
  = q^{\numvar-1} (q{-}1) r - q^{\numvar-1}r\degpol = q^{\numvar-1}r \cdot (q{-}1{-}\degpol)
\]
which meets the designed distance of the extended linearized
Reed--Muller code. In this case, the drop on the minimum distance
has then been absorbed.
\end{remark}

\subsection*{Acknowledgements}  This work was partially funded by the grants ANR-21-CE39-0009-BARRACUDA and ANR-22-CPJ2-0047-01. The first author thanks Cícero Carvalho for suggesting to work on Reed--Muller codes in the sum-rank metric, during the International Conference on Algebraic Geometry, Coding Theory and Combinatorics in honour of Sudhir Ghorpade. 

\bibliography{biblio_lrm}
\bibliographystyle{plain}
\end{document}